%% file: main.tex
\documentclass[11pt]{article}
\usepackage[normalem]{ulem}
\input{macros}

\title{Inference for Functional Data under Markov Constraints}

\author{Ulysse Naepels\thanks{ulysse.naepels@epfl.ch} }
\author{Victor M. Panaretos\thanks{victor.panaretos@epfl.ch}}

\affil{Institute of Mathematics, Ecole Polytechnique Fédérale de Lausanne}

\date{\today}

\begin{document}

\maketitle

\begin{abstract}
Smoothness has long been the dominant form of parsimony in functional data analysis, to the point of occasionally being conflated with the very notion of functional data. However, many core inferential tasks depend on the inverse covariance, where sparsity—rather than smoothness—emerges as the more natural structural constraint. In this paper, we explore Markovianity as an alternative to smoothness. Focusing on the Gaussian case as a central motivating setting, we exploit the fact that Markovianity induces a shape constraint on the covariance kernel.
Building on this observation, we introduce a Markov transform of the empirical covariance together with a corresponding estimator that enforces the Markov structure. The estimator is adaptive  and requires no regularity of the underlying covariance beyond continuity. In simulation experiments, it is seen to improve prediction performance even under model misspecification. Unlike smoothness-based assumptions, Markovianity is falsifiable. To assess its validity, we further propose a novel and computationally efficient test for the Markov property based on a new characterization of continuous graphical structure.
\end{abstract}

\newpage

\tableofcontents

\newpage

\section{Introduction}\label{sec:introduction}

\subsection{Parsimony beyond smoothness in functional data}

Functional data analysis (FDA) is  concerned with the problem of inference where data and estimands are elements of infinite dimensional function spaces -- typically random curves and their covariance kernels.  
Historically,  smoothness of the curves has been a prevailing assumption, in order to bridge the continuum nature of the data/estimands with the discrete and often sparse/noisy nature of the actual observations through smoothing  \citep{ramsay2005, yao2005functional}. 
Nevertheless, smoothness is not a defining property of functional data; rather, it is merely one particular form of parsimony that can be asserted.
It is rather the trace-class condition on the covariance operator that constitutes the genuine structural requirement separating functional from high-dimensional data -- reflecting the increasing correlation of nearby observations as the sampling resolution becomes finer \citep{hsing2015theoretical}.
For instance, diffusion processes provide a natural and core class of random functions that escape the framework of smoothness. Yet, despite being almost surely nowhere differentiable, their sample paths can still be regarded as functional observations as their covariances are traceable \citep{mohammadi2024rough}. 
If smoothness is not an intrinsic feature of functional data, it is natural to ask what other notions of parsimony could bridge finite/noisy observations and functional inferences 

Arguably, parsimony should ideally be imposed where it is most structurally relevant for inference, and this depends on the task at hand, for which covariance estimation is often but a preliminary step. Indeed, several downstream statistical tasks such as kriging, regression, prediction, and testing, involve the \emph{inverse} of the covariance rather than the covariance itself. In this context, 
smoothing adversely affects the task at hand, by increasing the ill-conditioning at the level of the inverse. It is therefore typically combined with additional layers of regularization, leading to the necessity to tune multiple smoothing parameters. When the inverse covariance is the primary object, thus, smoothing can amount to regularizing at the wrong level \citep{carroll2013unexpected}. At the level of inverse covariance, sparsity has been long recognized to be the natural parsimony-inducing assumption \citep{dempster1972covariance, yuan2007glasso, lam2009sparsistency}.
In high-dimensional setting sparsity directly targets variance control through the structure of the covariance inverse \citep{bickel2008regularized, liu2020cholesky, alghattas2024covariance}, and is interpretable as a structural assumption on the conditional dependence structure \citep{meinshausen2006graphs}. Within this perspective, the Markov property emerges as the sparsest nontrivial local dependence structure and thus a natural structural assumption for functional data.

Despite these developments in the high-dimensional case, the Markov assumption has not been systematically studied as a parsimony-inducing principle for functional inference. A notable exception is when functional data are censored, and Markovianity is introduced as a means of ensuring identifiability. For example, \citep{delaigle2016approximating} asserted the Markov property in order to carry out prediction from fragmented sample paths by way of non-parametric estimation of transition probabilities. A related approach, through diffusion modeling and bridges, was more recently proposed in the same fragment context by \citep{zhou2025dynamic}. These are special instances of the covariance completion problem, whose general solution -- and relation to Markov-type assumptions -- was obtained by \citep{waghmare2022completion}. From a complementary perspective, \citep{waghmare2025continuously} study the estimation of the Markov structure of functional data, when a consistent covariance estimator is already available.

\subsection{Contributions}

This work makes two main contributions: 
\begin{itemize}[nosep]
\item \textbf{Covariance estimation under the Markov constraint.} 
We introduce the \emph{Markov transform}, a projection of any covariance kernel onto the class of kernels corresponding to Gaussian Markov processes (\Cref{sec:markov-covariance-model}). 
This forms the basis of a non-parametric   covariance estimator (\Cref{sec:covariance-estimation}), conforming to the Markov structure while exploiting its inherent sparsity. 
We establish convergence rates across different sampling regimes (\Cref{sec:rates}), attaining parametric rates under dense designs, and nonparametric rates under sparse or irregular designs. 
In dense regimes, the estimator requires no tuning parameters, extending the paradigm of shape-constrained inference to functional data with Markov structure.

\item \textbf{Testing the Markov property.} 
Unlike smoothness assumptions, the Markov property is falsifiable, making it amenable to diagnostic testing. We provide a novel characterization of the Markov property for continuous processes under suitable identifiability conditions (\Cref{subsec:endpoints-characterization}). 
This structural result allows us to design a computationally efficient testing procedure (\Cref{subsec:test-via-conditional-correlations}) that allows for diagnostic checking of the Markov assumption. In particular, the procedure scales linearly with the number of observation points $p$, in contrast to the $O(p^3)$ complexity of standard conditional independence tests. 
Numerical experiments (\Cref{sec:numerics}) also demonstrate the practical performance of the test, confirming the computational gains and its power to detect departures from the Markov property.
\end{itemize}
Taken together, these contributions provide the basis for functional data analysis resting on Markov-based rather than smoothness-based parsimony. Before presenting our main results, we introduce the necessary background and notation (\Cref{subsec:background}).
All proofs and additional technical details are deferred to the appendix.

\subsection{Background and notation}\label{subsec:background}
\paragraph{Random elements of a Hilbert space.}
Let $\cH = L^2(I)$ be the Hilbert space of square-integrable 
real-valued functions defined on the compact interval $I$.
The space $\cH$ is equipped with the inner product
$\innerp{f}{g} = \int_I f(t)g(t)dt$ and corresponding norm $\norm{\cdot} = (\innerp{\cdot}{\cdot})^{1/2}$.
We consider random elements taking values in $\cH$. 
When such a random element $X$ has almost surely continuous paths, that is,
the map $X(w): I \to \RR$ is continuous for $\PP$-almost all $w \in \Omega$, it can equivalently be viewed as a stochastic process $\{X(t), t\in I\}$. 
Specifically, $X: I \times \Omega \to \RR$ is jointly measurable, where $(\Omega, \cF, \PP)$ is the underlying probability space. Functional data are typically assumed to be random elements of $\mathcal{H}$ with a.s.-continuous paths.
The mean and covariance functions of $X$ are defined as $\mu(t) = \EE\{X(t)\}$ and $K(s,t) = \cov\{X(s), X(t)\}$ for all $s,t\in I$, and with the latter inducing the covariance operator via
\begin{equation*}
    \bK: f \in L^2(I) \mapsto \int_I K(\cdot,t) f(t)dt \in L^2(I).
\end{equation*}
Provided $\mathbb{E}\|X\|^2<\infty$, the covariance operator is trace-class, 
and hence Hilbert-Schmidt \citep{hsing2015theoretical}. In practice, functional data are seldom observable in this idealized continuum and noiseless setting. Rather, one typically observes $n$ independent replications of the stochastic process $X$ at possibly random points, subject to measurement errors:
\begin{align}\label{eq:fda-setup}
    Y_{ij} = X_{i}(t_{ij}) + u_{ij} \quad \text{for } i \in \{1, \dots, n\}, \; j \in \{1, \dots, r\},
\end{align}
where $r \geq 2$ may grow with $n$, and $(u_{ij})_{i,j}$ are i.i.d. centered random variable with finite variance $\nu^2$, independent of the sequence $(X_i)_{i \in \NN}$.
We distinguish two regimes of observations.
First, the regular-observations regime: variables $(t_{ij})_{i,j}$ are deterministic and form a regular grid on the interval $I$. In particular, the observation points do not depend on the sample.
In this setting, $r = r_n$ may grow with $n$.
Second, the irregular regime, where $(t_{ij})_{i,j}$ are i.i.d. uniform random variables on $I$, independent of all other random quantities. 
Here, $r$ may grow with $n$ or remain bounded -- the latter case is referred to as the sparse regime. A core problem is, subsequently, to recover the underlying mean and covariance functions.

\paragraph{Gaussian processes.}
The class of Gaussian functional data is of particular interest in this study.
A stochastic process $\{X(t), t \in I\}$ is said to be Gaussian if for every $n \geq 1$, $(a_i)_{1 \leq i \leq n} \in \RR^n$ and $(t_i)_{1 \leq i \leq n} \in I^n$, the random variable given by the linear combination $\sum_{i=1}^n a_i X(t_i)$ is Gaussian.
This extends the notion of Gaussian vector to infinite index sets.
For Gaussian variables, independence and
uncorrelatedness are equivalent. Therefore,
the independence structure of a Gaussian process is
fully encoded by the zeros of its covariance kernel.
Moreover, if $(Y_1, Y_2, Y_3)$ is a Gaussian vector,
the conditional independence $Y_1 \ind Y_2 \mid Y_3$
holds if and only if $\cov(Y_1, Y_2 \mid Y_3) =0$ almost surely.

\paragraph{Graphical models.}
The conditional independence structure of a collection of random variables can be encoded by a graph, which provides a concise representation of the interactions among the variables. 
An undirected graph is a pair $\mathcal{G} = (V, E)$ where $V$ is a set of vertices and 
$E \subset V \times V$ is a set of edges between the vertices with $(v, v) \notin E$ for all $v$.
The conditional independence graph of a distribution $P$ on $\RR^p$ is the undirected graph $\mathcal{G} = (V, E)$ where given $Z \sim P$, 
\begin{equation*}
    \{ j,k \} \in E \iff Z_j \notind Z_k \mid Z_{-jk}.
\end{equation*}
Furthermore, if $Z$ is Gaussian with full-rank covariance $C$, edges in the graph correspond to non-zero entries in the precision matrix: $\{ j,k \} \in E\mid Z_{-jk} \iff (C^{-1})_{i,j} \neq 0$.

We also define graphical models on continuous sets and the graph of a stochastic process.
Let $X = \{X_t \: : \: t\in I\}$ be a collection of random variables where $I$ is a continuous index set and denote $X_W = \{X_t \: : \: t\in W\}$ the restriction of the process to the index set $W \subset I$. 
We say that $(X, \Omega)$ is a graphical model if  $X_s \ind X_t \mid X_W$ for all $s, t$ separated by $W$ in $\Omega$.
We call the graph of the process $X$ the intersection of all closed graphs $\Omega$ for which $(X, \Omega)$ is a graphical model.
See \Cref{apx:graphical-models} for more details.
 When the index set is uncountable, a precision-based characterization is precluded by the absence of an operational notion of inverse of a kernel defined on an uncountable space.
It is possible to circumvent the issue by approximating the true graph.
Let $\pi = \{I_j\}_{j\in[m]}$ be a partition of $I$ into $m$ intervals of equal length,
and let $K$ be a covariance kernel defined on $I \times I$. 
We write for $i,j \in [m]$, $K_{ij} = K|_{I_i \times I_j}$ and the corresponding integral operator $\bK_{ij}$.
By \citep{baker1973joint}, there is, for $i \neq j$, a unique bounded linear operator $\bR_{ij}$, referred to as the \emph{cross-correlation operator}, such that 
$\bK_{ij} = \bK_{ii} \bR_{ij} \bK_{jj}$ and $\norm{\bR_{ij}} \leq 1$.
When $\bR_\pi \coloneqq (\bR_{ij})_{i,j\in[m]}$ is invertible, we can define a precision operator matrix $\bP_\pi \coloneqq \bR_\pi^{-1}$ which can be used to approximate the graph of a process under suitable conditions (see \Cref{sec:test-markov} and \citep{waghmare2025continuously}).

\section{The Markov property as a shape-constrained covariance model}\label{sec:markov-covariance-model}

In the functional setting, there is a core obstruction in the sparse modeling of the inverse covariance. Namely, owing to its trace-class nature, the covariance operator is not boundedly invertible;  therefore, its inverse is not tangible as an object that can be directly modeled or manipulated. Fortunately, for Gaussian processes, the Markov property yields an explicit factorization of the covariance structure and thereby induces a shape-restricted model for the covariance kernel. The induced shape constraint
parallels the multivariate setting, where Gaussian Markov vectors have tridiagonal precision matrices, but phrasing the constraint at the level of covariance allows for its statistical operationalization in the functional context.

Assuming Gaussianity amounts to adopting a working model for second-order structure, that is, restricting attention to distributions determined by their first two moments. At the level of the Markov property, this can be interpreted as a second-order approximation to the conditional independence structure. In the non-Gaussian case, the same structure can be interpreted as a conditional uncorrelatedness model. For ease and clarity, we will focus on the canonical case of Gaussians.

Consider a Gaussian process $X$ defined on a compact interval $I$ of $\RR$. 
We assume that $X$ has mean zero and continuous covariance $K$ on $I \times I$, and is nowhere degenerate except possibly at the endpoints of $I$.
Let $(\mathcal{F}_t, \: t \in I)$ be a filtration adapted to $X$.
The process $X$ is said to satisfy the Markov property if
$\PP(X_t \in A \mid \mathcal{F}_s) = \PP(X_t \in A \mid X_s)$
for all $s<t$ and for every Borel set $A$.
Equivalently, $X$ satisfies the conditional independence property
$X_s \ind X_t \mid X_u $ for all  $s \leq u \leq t$.
For Gaussian processes, conditional independence is equivalent to conditional uncorrelatedness. 
Hence the Markov property can be written in terms of covariance as
$\cov(X_s, X_t \mid X_u) = 0$ for all $s \leq u \leq t$.
Moreover, the conditional covariance of a Gaussian process can be expressed as:
\begin{equation*}
    \cov(X_s, X_t \mid X_u) = K(s,t) - \frac{K(s, u) K(u,t)}{K(u,u)}.
\end{equation*}
This provides a core characterization of the Markov property entirely in terms of the covariance kernel.
\begin{proposition}[\citep{doob1962stochastic}]\label{prop:kernel-functional-equation}
    $X$ satisfies the Markov property if and only if its covariance
    kernel verifies the following equation:
    \begin{equation*}
        K(s,t) = \frac{K(s,u)K(u,t)}{K(u,u)} \quad \text{for all } s \leq u \leq t.
    \end{equation*}
\end{proposition}
\noindent The displayed characterization is equivalent to 
\begin{equation}\label{eq:kernel-decomp}
    K(s, t) = f(s) g(t), \quad \text{for} \; s \leq t,
\end{equation}
where $f$ and $g$ are two continuous functions defined on $I$ \citep{mehr1965}.
This representation implies a separable structure along each coordinate direction. Consequently, the horizontal (resp. vertical) level sets coincide up to a multiplicative constant (see \Cref{fig:kernel-decomp}).
\begin{figure}[h!]
    \centering
\begin{tikzpicture}
    \draw[red, line width=2pt] (0,0) -- (0,3);
    \draw[blue, line width=2pt] (0,3) -- (3,3); 
    \draw[line width=1pt] (3,3) -- (3,0);
    \draw[line width=1pt] (3,0) -- (0,0); 

    \draw[line width=1.5pt] (0,0) -- (3,3);

    \begin{scope}
        \clip (0,0) -- (3,3) -- (0,3) -- cycle;

        \draw[step=0.1cm, red!60, very thin] (0,0) grid[xstep=0.25cm, ystep=100cm] (3,3);

        \draw[step=0.1cm, blue!60, very thin] (0,0) grid[xstep=100cm, ystep=0.25cm] (3,3);
    \end{scope}

    \node[blue] at (1.5, 3.3) {$f$}; 
    \node[red] at (-0.2, 1.5) {$g$};

    \node at (1.1, 2.1) {$f(s)g(t)$};
\end{tikzpicture}

    \centering
    \caption{Covariance of a Gaussian Markov process.}
    \label{fig:kernel-decomp}
\end{figure}
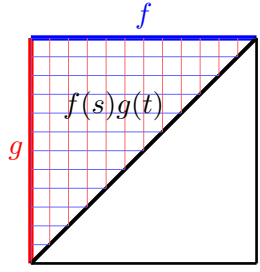

Thus, the covariance kernel on the upper triangle $s\leq t$ is fully determined by two one-dimensional functions. 
In this sense, a Markov covariance kernel is effectively a one-dimensional object embedded in a two-dimensional domain. 
Parsimony is therefore achieved through structural constraints rather than smoothness assumptions. 
This allows for inference driven by structure, rather than by regularity considerations.

\section{Covariance estimation under Markov constraints}\label{sec:covariance-estimation}

As a first step, we consider the problem of estimating the covariance based on $n$ independent realizations of a stochastic process $X$ observed without error at discrete design points. This will serve as the basis for our more general methodology. 
Our estimation strategy is guided jointly by the structure of the estimand and the discrete observation scheme, and relies on a projection principle induced by the Markov structure. 
Specifically, we first estimate the covariance on the observation grid using the Markov factorization, yielding a nodal estimator that naturally aligns with a graphical-model perspective. 
We then propagate this estimator to the continuum via linear interpolation.

\subsection{The Markov transform}
We start by defining the \emph{Markov transform} of a covariance.
Formally, let $X$ be a Gaussian process defined on the compact interval $I$. Assume $X$ is nowhere degenerate, has mean zero, and has a continuous covariance kernel $K$.
Let $(t_1, \dots, t_p) \in I^{p}$ be an increasing sequence of distinct design points. Since $X$ is a Gaussian process, the pair $(X(t_j), X(t_{j+1}))$ admits a linear regression representation: there exists $\beta_{j+1}$ such that
\begin{equation}\label{eq:gaussian-regression}
    X(t_{j+1}) = X(t_{j}) \beta_{j+1} + \epsilon_{j+1}
    \quad \text{and } \beta_{j+1} = \frac{\cov(X(t_j), X(t_{j+1}))}{\var(X(t_j))},
\end{equation}
where $\epsilon_{j+1}$ follows a normal distribution with variance 
$\sigma_{j+1}^2 = K(t_{j+1}, t_{j+1}) - K(t_{j}, t_{j+1})^2 / K(t_{j}, t_{j})$ and mean zero. If $X$ satisfies the Markov property, \Cref{prop:kernel-functional-equation} implies that

\begin{align}\label{eq:markov-transform}
    K(t_j, t_k) =  K(t_j, t_j) \cdot \prod_{l=1}^{k-j} \beta_{j+l} \eqqcolon K^M(t_j, t_k) \quad \text{for all } j \leq k.
\end{align}
We refer to the matrix defined by the right-hand side of \eqref{eq:markov-transform} as the \emph{Markov transform}. If $X$ is not Markov, the Markov transform $K^M$ can still be defined, but $K \neq K^M$.
The following proposition establishes that the Markov transform is the information projection onto the class of Markov covariances. Specifically, it minimizes the Kullback-Leibler divergence to the class of Gaussian distributions whose covariances have tridiagonal inverses.
\begin{proposition}\label{prop:markov-projection}
    Let $X = \{X(t), t \in I\}$ be a {non-degenerate} Gaussian process with mean zero, and let $K_p$ be the covariance matrix of $(X(t_1), \dots, X(t_p))$.
    The following statements are equivalent:
    \begin{itemize}[nosep]
        \item[(1)] $K_p^M$ is the Markov transform \eqref{eq:markov-transform} of $K_p$,
        \item[(2)] $K_p^M = \argmin_{\Sigma \in \mathcal{M}} KL(\mathcal{N}(0, K_p), \mathcal{N}(0, \Sigma))$, 
        where $\mathcal{M} = \{\text{PD matrices with tridiagonal inverse}\}$,
        \item[(3)] $K_p^M$ is the covariance of  $(Z(t_1), \dots, Z(t_p))$ where 
        $Z(t_{j+1}) = \beta_{j+1} Z(t_j) + \varepsilon(t_{j+1})$, 
        $Z(t_1) \sim \mathcal{N}(0, K_{p,11})$, and
        $\{\varepsilon(t_{j+1})\}_{j\in[p-1]}$ are independent centered Gaussian variables with variance $\var(X(t_{j+1})) - \beta_{j+1}^2 \var(X(t_j))$.
    \end{itemize}
\end{proposition}
By (3), the Markov transform yields the Gaussian $AR(1)$ process that is closest to the original Gaussian process in information distance. In general, this $AR(1)$ process is non-stationary.\begin{remark}
    The Markov transform can further be interpreted as the covariance obtained by completing the matrix $\{K(t_j, t_k)\}_{|k-j| \leq 1}$ while enforcing a tridiagonal precision matrix. 
    This defines a global completion rule rather than a local or greedy procedure: the entire covariance matrix is determined consistently with the Markov structure. 
    In fact, it can be shown that this construction is equivalent to the canonical completion introduced in \citep{waghmare2022completion}, which extends a partially observed covariance kernel by fitting the minimal compatible graphical model.
\end{remark}

\subsection{The Markov estimator under synchronous and noise-free observations}
When the functional data are observed without measurement error on the same grid, the Markov transform suggests the following estimator on the observation grid:
\begin{align*}
    \KhatM_p(t_j, t_k) = \hat{\sigma}^2_j \cdot \prod_{l=1}^{k-j} \hat{\beta}_{j+l} \quad \text{for } j \leq k
\end{align*}
where $\hat{\beta}_{j+1}$ is obtained by regressing $X(t_{j+1})$ onto $X(t_{j})$ and $\hat{\sigma}^2_j$ is the unbiased estimator of the variance of $X(t_{j})$.
The estimator $\KhatM_p$ naturally enforces the Markov property during the estimation: the Markov factorization \Cref{prop:kernel-functional-equation} holds at every step, ensuring that the covariance estimator corresponds to a discrete-level Markov process.

We then propagate the estimator defined on the observation grid to construct a genuine functional object. 
A natural approach is to linearly interpolate the nodal Markov estimator.
By exploiting the shape restriction imposed by the Markov assumption, this approach offers both methodological and theoretical advantages: no tuning parameters are required, and the resulting surface inherently respects the shape constraint and continuity.

\begin{remark}
    Although our method does not require any smoothness assumptions, it can accommodate additional smoothing if desired. 
    One could smooth the covariance over the upper triangle \citep{mohammadi2024rough}. 
    In the case of Markov processes, the covariance factorization allows smoothing only the two univariate functions defining the kernel, significantly reducing computational complexity while preserving the Markov structure.
    However, this re-introduces bandwidth selection and implicit smoothness assumptions, and diminishes the advantage of the shape-restricted approach. 
\end{remark}

\subsection{The Markov estimator under irregular  and noisy observations}\label{subsec:markov-estimator-irregular}

Assume now that we observe each of $n$ independent realizations
$X_1, \dots, X_n$ at $r$ i.i.d. random times
$(t_{ij})_{i \in [n], j \in [r]}$, under measurement errors:
\begin{align}\label{irregular/noisy_model}
    Y_{ij} = X_{i}(t_{ij}) + u_{ij} \quad \text{for } i \in \{1, \dots, n\}, \; j \in \{1, \dots, r\}.
\end{align}
The errors $\{u_{ij}\}_{i,j}$ are i.i.d. centered random variables with finite variance $\nu^2$, such that the sequences $(X_i)_{i \in \NN}$, $(u_{ij})_{i,j}$ and $(t_{ij})_{i,j}$ are independent. The number of per-curve evaluations $r$ can grow with $n$ (dense case) or remain bounded (sparse case).

The asynchronous nature of the design does immediately lend itself to the definition of the regression coefficient estimators as in the previous paragraph. Nevertheless, we now show that the asynchronous setting can be reduced to a synchronous grid setting. Partition the interval $I$ into  equally sized and ordered sub-intervals $B_1, \dots, B_p$.
For each sample $i$, let
\begin{align*}
    \cJ_i(k) = \left\{ j \in [r]: \; t_{ij} \in B_k \right\} \quad \text{and} \quad \cI_n(k) = \left\{ i \in [n]: \; \exists (j, j') \text{ s.t. } j \in \cJ_i(k), j' \in \cJ_i(k+1) \right\}.
\end{align*}
The regular-observations regime corresponds to
the case where {$r = p$}, $\cJ_i(k) = \{t_k\}$ and $\cI_n(k) = \{1, \dots, n\}$.
Denote $m_i(k) =  | \cJ_i(k) |$ and $n_k = | \cI_n(k) |$.
We then define the regression coefficient between consecutive sub-intervals $B_k$ and $B_{k+1}$ 
\begin{align*}
    \hat{\beta}_{k+1} = \frac{1}{\hat{\sigma}_k}\sum_{i \in \cI_n(k)} \left( \frac{1}{m_i(k) m_i(k+1)} \sum_{(j, j') \in \cJ_i(k) \times \cJ_i(k+1)} Y_{ij} Y_{ij'} \right),
\end{align*}
where $\hat{\sigma}_k$ is the de-noised variance estimator for the bin $B_k$
\begin{align*}
    \hat{\sigma}^2_k = \frac{1}{n_k} \sum_{i \in \cI_n(k)} \left[ \left( \frac{1}{m_i(k)} \sum_{j \in \mathcal{J}_i(k)} Y_{ij} \right)^2 - \frac{\nu^2}{m_i(k)} \right],
\end{align*}
and the noise error $\nu^2$ can be estimated by taking the difference between the average of the squares and the average of the cross-products within the bins. By the mean-value theorem and as the covariance $K$ is continuous, there exist $t^*_k \in B_k$ and $t^*_{k+1} \in B_{k+1}$ such that $\hat{\beta}_{k+1}$ consistently estimates $R(t^*_k, t^*_{k+1})$ with convergence rate of order $1 / \sqrt{n_k}$. 

For $s, t \in I$, let $j, k$ and $\alpha_s, \alpha_t$ be the indices and weights such that $s = (1-\alpha_s)t^*_j + \alpha_s t^*_{j+1}$ and $t = (1-\alpha_t)t^*_k + \alpha_t t^*_{k+1}$.
Then, the Markov estimator is defined as the bilinear interpolant of $\KhatM_p$, namely
\begin{align}\label{eq:markov-estimator}
    \KhatM(s, t) = \sum_{a,b \in \{0,1\}} w_a(s) w_b(t) [\KhatM_p]_{j+a, k+b}
    \quad \text{where }
    [\KhatM_p]_{j+a, k+b} = \hat{\sigma}^2_{j+a}  \prod_{l=j+a}^{k+b-1} \hat{\beta}_{l+1} 
\end{align}
with $w_0(u) = (1 - \alpha_{u})$ and $w_1(u) = \alpha_{u}$ for $u \in I$.

\section{Rates of convergence}\label{sec:rates}

In light of the link between the asynchronous and synchronous designs, we will present  convergence rates in the more general observation setting given in Eq. \eqref{irregular/noisy_model}.  
This will allow for an appreciation of the impact of the various parameters at play.

\subsection{Mean squared error}
The analysis of the mean squared error of the Markov estimator $\KhatM$ \eqref{eq:markov-estimator} is driven by two key observations.
First, the empirical regression coefficients used to construct the estimator 
are uncorrelated due to the Markov property and the independence of the measurement errors.
Second, the continuity of the covariance kernel enables the control of the true regression coefficients  as the number of distinct sampling locations $\{t_{ij}\}$ grows.
Together, these two results allows for the following upper bound to be established.

\begin{theorem}\label{thm:continuous-rates} 
Let $X$ be a mean-zero Gaussian Markov process on a compact interval $I$ with covariance $K$. 
Assume that $X$ is nowhere degenerate and that $K$ is continuous with modulus of continuity $w$.
Under the irregular and noisy observation regime \eqref{irregular/noisy_model} with $p$ bins as defined in \Cref{subsec:markov-estimator-irregular}, if $p/n$ is bounded, then the Markov estimator $\KhatM$ satisfies
\begin{align*}
    \left(\EE \left[ \norm{\KhatM - K}_2^2 \right]\right)^{1/2} 
    \leq 
    M \Bigg[
        \frac{p}{n} \left(1 \vee \frac{p^2}{r^2}\right)
        \Big( w\!\left(\frac{2}{p}\right) + C \nu^2 \Big)
    \Bigg]^{1/2}
    + \frac{L}{\sqrt{n}}
    + w\!\left(\frac{2}{p}\right),
\end{align*}
where $M,C,L>0$ are constants independent of $n$, $r$, and $p$, and $\nu^2$ denotes the variance of the measurement errors.
\end{theorem}

The error bound reflects a bias-variance tradeoff: stochastic error arises from estimating the nodal coefficients, while the deterministic approximation is controlled by the modulus of continuity $w(\cdot)$. 
In dense designs ($r$ growing sufficiently fast relative to $n$), the variance is limited by the number of samples. 
Lipschitz regularity yields the parametric rate $n^{-1/2}$ without measurement error, whereas an $\alpha$-Hölder modulus gives $n^{-\alpha/(1+\alpha)}$. 
Measurement error ($\nu>0$) reduces the achievable rate; for Lipschitz kernels, convergence slows to $n^{-1/3}$. 
In sparse designs ($r$ bounded), limited intra-curve information shifts the dominant variance term, and consistency requires $p$ to diverge with $n$, attaining $n^{1/(3+\alpha)}$ in the noiseless case and $n^{1/(3+2\alpha)}$ with noise. This yields rates of $n^{-1/4}$ (noiseless) and $n^{-1/5}$ (noisy) for Lipschitz kernels. In intermediate regimes ($r \asymp n^\delta$), rates interpolate between sparse and dense extremes, with phase transitions at $r \asymp n^{1/2}$ (noiseless) and $r \asymp n^{1/3}$ (noisy).

\begin{table}[ht]
\centering
\renewcommand{\arraystretch}{1.2}
\label{tab:rates}
\begin{tabular}{llcc}
\toprule
Design & Noise & $\alpha$-Hölder rate & Lipschitz rate \\
\midrule

\multirow{2}{*}{Dense}
& $\nu>0$ 
& $n^{-\alpha/(1+2\alpha)}$ 
& $n^{-1/3}$ \\

& $\nu=0$ 
& $n^{-\alpha/(1+\alpha)}$ 
& $n^{-1/2}$ \\

\midrule

\multirow{2}{*}{Mixed}
& $\nu>0$ 
& $(nr^2)^{-\alpha/(3+2\alpha)}$ if $r \lesssim n^{1/(1+2\alpha)}$

& $(nr^2)^{-1/4}$ if $r \lesssim n^{1/3}$ \\

& $\nu=0$ 
& $(nr^2)^{-\alpha/(3+\alpha)}$ if $r \lesssim n^{1/(1+\alpha)}$

& $(nr^2)^{-1/3}$ if $r \lesssim n^{1/2}$ \\

\midrule

\multirow{2}{*}{Sparse ($r$ bounded)}
& $\nu>0$ 
& $n^{-\alpha/(3+2\alpha)}$ 
& $n^{-1/5}$ \\

& $\nu=0$ 
& $n^{-\alpha/(3+\alpha)}$ 
& $n^{-1/4}$ \\

\bottomrule
\end{tabular}
\caption{Optimized convergence rates for the Markov estimator.}
\end{table}

\begin{remark}
    Interestingly, in regimes where the sampling density $r$ scales at the same order as the optimal resolution (e.g., $r \asymp n^{1/3}$ in the noisy Lipschitz case), the estimator can be implemented in a completely tuning-free manner by setting $p=r$. In such cases, the sampling design itself provides the necessary regularization, and the Markov estimator automatically recovers the best possible rates without external parameter selection.
\end{remark}

\subsection{Model misspecification and bias quantification}
When the data do not satisfy the Markov property, we can quantify the deviation from the Markov model, and use this to quantify the bias incurred by working with the Markov assumption.  Let $X$ be a Gaussian process on $I$ with mean zero and covariance kernel $K$, and let $\{t_1, \dots, t_p\}$ be a grid on $I$.
Define $K_p \coloneqq (K(t_j, t_k))_{j,k \in [p]}$, and let $K^M_p$ denote the Markov transform \eqref{eq:markov-transform} of $K_p$.
Assume we observe $n$ independent replications of $X$, and denote the Markov estimator by $\KhatM_p$.
Then, by the triangle inequality,
\begin{align*}
    \norm{\KhatM_p - K_p}_{F}^{2} \leq \norm{\KhatM_p - K^M_p}_{F}^{2} + \norm{K^M_p - K_p}_{F}^{2}.
\end{align*}
The first term is upper bounded by the stochastic error obtained in \Cref{thm:continuous-rates} up to a factor of $p^2$.
The second term can be expressed as a deviation from a multiplicative correlation structure.
For $j,k \in [p]$ denote by $\rho_{j,k}$ the correlation between $X(t_j)$ and $X(t_k)$, that is $\rho_{j,k} = K_{j,k}/(K_{j,j}K_{k,k})^{1/2}$.
Then,
\begin{align}\label{eq:misspecification}
    \norm{K_p-K^M_p}_{F}^{2} = 2 \sum_{j<k} K_{jj} K_{kk}\left(\rho_{j,k} - \prod_{l=j+1}^{k}\rho_{l,l+1} \right)^2.
\end{align}
The multiplicative correlation structure corresponds to the functional equation characterizing the Markov property for Gaussian processes (\Cref{prop:kernel-functional-equation}). Normalizing by the Frobenius norm yields a practical measure of goodness-of-fit. We note that all terms on the right-hand side, as well as the Frobenius norm itself, are estimable. In the next section, we develop a more formal criterion for goodness-of-fit in the form of hypothesis testing.

\section{Testing the Markov property}\label{sec:test-markov}

The Markov property is a conditional independence constraint and is therefore testable. 
Hence, it represents a form of parsimony that is, in principle, falsifiable, unlike smoothness which is generally not a testable property \citep{devroye1999almost}.
In particular, misspecification of the Markov model can be detected. 
For instance \citep{chen2012testing, zhou2023testing} test for the Markov property in the stationary case via the frequency domain, using conditional characteristic functions. In a general, non-stationary setting such as ours, however, there is an a priori large number of conditional dependencies to be tested -- indeed growing as the sampling grid is refined.

We propose a novel characterization of the Markov property, which we refer to as the ''endpoint characterization''. 
This result enables the construction of a testing procedure that circumvents the substantial computational burden typically associated with verifying the Markov property. 
The endpoint characterization drastically reduces the required number of conditional independence tests.
We develop theoretical results for Gaussian processes, but our approach can be extended to second order stochastic processes by replacing conditional independence by 'conditional uncorrelatedness'.

\subsection{Endpoint characterization of Markov processes}\label{subsec:endpoints-characterization}

We consider a stochastic process $X$ defined on an interval $[a,b]$ with almost-surely continuous sample paths and which nowhere degenerate.
Our goal is to show that the Markov property of $X$ can be characterized solely through conditional independencies between the endpoints $a$ and $b$ given any intermediate point. 
However, unlike in the finite-dimensional setting, certain types of degeneracy can occur in infinite dimensions.
Specifically, degeneracies may occur at a global scale, such as linear dependencies across disjoint branches of the process, or at a local scale, for instance through perfect predictability of future values from arbitrarily close past observations. 
To exclude these situations, we introduce the notions of \emph{non-singularity of the correlation structure} and \emph{strong local non-determinism}.

\begin{definition}\label{def:non-singular-correlation}
    The correlation structure of $X$ is said to be non-singular if the correlation operator matrix $\bR_{\pi_j}$ defined in \Cref{subsec:background} is invertible for every $j \in \NN^{*}$, where $\pi_j$ denotes the partition of $I$ into $j$ equal-length intervals.
\end{definition}
Intuitively, non-singularity of the correlation structure asserts that there can be no exact linear dependencies across disjoint branches of the process.  
It implies the identifiability of the graph of the process \citep{waghmare2025continuously}, and that the graph fully characterizes the conditional independence structure of the process (\Cref{thm:non-singular-correlation}).
 This condition is intrinsic and does not depend on the choice of design grid: any sequence of partitions separating points on $I$ yields the same property. In finite dimensions, the condition is equivalent to having a non-singular correlation matrix.

\begin{definition}
    We say that a stochastic process $X$ in $L^2(I)$ 
    is strongly locally nondeterministic if there is a function 
    $\phi: \RR_+ \to \RR_+$ such that $\phi(0) = 0$, $\phi(r) > 0$ for $r > 0$, and
    \begin{align*}
        \var(X(t) \mid X_{s} : r \leq |s-t| \leq r_0) \geq c_0 \phi(r) 
        \quad \text{for all } r \in (0, t \wedge r_0], \text{ and all } t\in I,
    \end{align*}
    where $c_0, r_0$ are positive constants.
\end{definition}
This definition formalizes the idea that future behavior cannot be predicted arbitrarily well
from nearby past values and quantifies the degree of unpredictability \citep{berman1973local, cuzick1982joint, xiao2007strong}.

\begin{theorem}\label{thm:endpoints-cond-ind}
    Let $X$ be a real-valued Gaussian process defined on a compact interval $I$ with continuous sample paths and nowhere degenerate.
    Assume that $X$ has a non-singular correlation structure and is strongly locally nondeterministic on $[a,b]$ where $a, b \in I$ with $a < b$.
    Then, $X$ is Markov on the interval $[a, b]$ if and only if 
    \begin{equation}\label{eq:endpoints-cond-ind}
        X_a \ind X_b \mid X_c \quad \text{for all } c\in (a,b).
    \end{equation}
\end{theorem}
The forward implication is immediate: the Markov property implies \Cref{eq:endpoints-cond-ind}. For the converse, we show that $X_s \notind X_t \mid X_u$ induces a path in $\Omega$ that connects a point before $u$ to a point after $u$. Such a path can then be used to construct a connection between $a$ and $b$ that avoids $u$. We conclude using non-singularity of the correlation structure and strong local nondeterminism.

\begin{remark}
    In practice, testing the Markov property of a continuous process observed at $p$ points requires performing $O(p^3)$ conditional independence tests. \Cref{thm:endpoints-cond-ind} shows that, under continuity and non-degeneracy assumptions, only $O(p)$ tests are needed.
\end{remark}

\subsection{Implementation via conditional correlations}\label{subsec:test-via-conditional-correlations}

Let $\cE([a,b])$ be the set of Gaussian distributions on $L^2([a,b])$ that have non-singular correlation structure and are strongly locally nondeterministic, and  $\cM_0$ the subset of distributions that satisfy the Markov property. Define $\cM_1 = \cE \setminus \cM_0$, the set of distributions
corresponding to the alternative hypothesis.
\Cref{thm:endpoints-cond-ind} implies that
\begin{align*}
    {\cM}_0 &= \left\{ P \in {\cE}([a,b]) : \; X \sim P \implies \forall c \in (a,b), \; X_a \ind X_b \mid X_c \right\}, \\
    {\cM}_1 &= \left\{ P \in {\cE}([a,b]) : \; X \sim P \implies  \exists c \in (a,b), \; X_a \notind X_b \mid X_c \right\}.
\end{align*}

Consider a process $X$ with distribution $P \in {\cE}([a,b])$
and suppose that we observe $n$ i.i.d.\ copies 
$(X^{(1)}, \dots, X^{(n)})$ of $X$ at times $\{t_1, \dots, t_p\}$.
We aim to test 
$X(t_1) \ind X(t_p) \mid X(t_j)$ for $j \in \{2, \dots, p-1\}$.
To do so, we define the empirical partial correlation
\begin{equation}\label{eq:empirical-correlation}
    \hat{\rho}_{j} = 
    \frac{
    \frac{1}{n}\sum_{k=1}^n e_{1j}^{(k)} e_{pj}^{(k)}
    }{
    \Big[\frac{1}{n}\sum_{k=1}^n (e_{1j}^{(k)})^2 
    \cdot 
    \frac{1}{n}\sum_{k=1}^n (e_{pj}^{(k)})^2 
    \Big]^{1/2}},
\end{equation}
where $e_{1j}$ (resp.\ $e_{pj}$) denotes the residual obtained by regressing $X(t_1)$ (resp.\ $X(t_p)$) on $X(t_j)$. Define the corresponding Fisher $z$-transform
\begin{align*}
    \hat{z}_j = \frac{1}{2}\log\left(\frac{1 + \hat{\rho}_j}{1 - \hat{\rho}_j}\right).
\end{align*}
Each $\hat{z}_{j}$ is approximately normally distributed with variance $1/(n-4)$.
Therefore, we reject the hypothesis $X_1 \ind X_p \mid X_j$ when:
$\sqrt{n-4} \left| \hat{z}_j \right| > \Phi^{-1}\left(1- \frac{\alpha}{2}\right)$,
where $\Phi$ is the cumulative distribution function of
the standard normal distribution. By the multivariate central limit theorem,
$\sqrt{n} (\hat{z}_2, \dots, \hat{z}_{p-1})^\top  \xrightarrow[]{} \cN_{p-2} \left( 0, \Sigma \right)$,
where $\Sigma = \lim_n n \cov(\hat{\mathbf{z}})$ and $\hat{\mathbf{z}} = (\hat{z}_2, \dots, \hat{z}_{p-1})^\top$.
We now define the aggregated test as
\begin{align}
    \widehat{T}_n = \max_{j \in \{2, \dots, p-1\}} \sqrt{n} \cdot |\hat{z}_j |.
\end{align}

The asymptotic distribution of $\widehat{T}_n$ can then be computed by 
using the continuous mapping theorem: for fixed $p$,
\begin{align*}
    \widehat{T}_n \xrightarrow[]{d} T \coloneqq \max_{j=2, \dots, p-1} |Z_j| 
    \quad \text{where } Z \sim \cN_{p-2}(0, \Sigma),
\end{align*}
and the limiting distribution is given by 
$\PP(T \leq t) = \int_{u \in [-t, t]^{p-2}} \varphi_{\Sigma}(u) du$,
where $\varphi_{\Sigma}$ is the density of the centered multivariate normal 
with covariance $\Sigma$.
The results of \citep{chernozhukov2013gaussian}
on Gaussian approximations allow us to derive an upper-bound
for the Kolmogorov distance between the proposed test
and the maximum of the Gaussian vector with corresponding
covariance.

\begin{theorem}\label{thm:test-rate}
    Let $\widehat{Y} \in \RR^{p-2}$ be a random vector with distribution $\cN(0, \widehat{\Sigma})$
    where $\widehat{\Sigma} = \cov \left( \hat{\mathbf{z}} \right)$.
    Then, if $p \lesssim \exp(n^{\alpha/7})$ for some $\alpha \in (0,1)$,
    we have
    \begin{align*}
        \sup_{t \geq 0} \left| \PP \left\{\widehat{T}_n \leq t \right\} - \PP \left\{ \max_{1<j<p} |\widehat{Y}_j| \leq t \right\} \right|
        \lesssim  \frac{1}{n^{1/8}} \left\{ \log(pn^{3/2}) \right\}^{7/8}.
    \end{align*}
\end{theorem}

The proof consists in defining a modification of the test that is at small distance from the original statistic and applying the results of \citep{chernozhukov2013gaussian} to this modification.
This allows for a broad class of high-dimensional settings, as $p$ can grow exponentially in $n^{\alpha}$ for $\alpha < 1/7$.

\begin{remark}
    Other aggregations than the maximum can be used to construct a final test statistic.
    For instance, the $L^2$ norm is well suited to detect many small deviations from the null, but has little power when only a few conditional independencies fail.
    An advantage of the maximum is that it controls the bias of $\widehat{T}_n$ by the largest bias among the $\hat{z}_j$, $j \in [p]$, whereas sum-based norms typically yield larger bias \citep{shah2018hardness}.
\end{remark}

\begin{remark} 
The test extends to the irregular observation setting. Empirical partial correlations can be computed from regression residuals obtained by regressing one process value onto another. As described in \Cref{subsec:markov-estimator-irregular}, regression coefficients can be estimated in the irregular regime, and the same approach applies to residuals. In this case, the procedure amounts to testing conditional independence across increasingly fine blocks. 
\end{remark}

\section{Numerical illustrations}\label{sec:numerics}

We now investigate the empirical behavior of the Markov estimator and the finite-sample performance of the proposed test. To that end, we consider the following stochastic processes, each sampled at $p = 20$ locations in $[0,1]$:
\begin{enumerate}[nosep]
    \item \textbf{Brownian motion}: mean-zero Gaussian Markov process with covariance kernel $K(s,t) = \min(s,t)$;
    \item \textbf{Stationary Ornstein--Uhlenbeck process}: Gaussian Markov process with covariance kernel $K(s,t) = (\sigma^2 / (2\theta)) \exp(-\theta |s-t|)$; in our experiments, we set $\theta = \sigma = 1$;
    \item \textbf{Kernel-embedded Brownian motion (KEBM)}: mean-zero Gaussian process obtained by embedding the Brownian motion covariance with a smoothing kernel of bandwidth $h$; in our experiments, we use a Wendland kernel with $h \in \{0.05, 0.1\}$. See \Cref{apx:numerics} for details.
\end{enumerate}
The Brownian motion and Ornstein-Uhlenbeck processes are used to evaluate performance within the Markov class, while the kernel-embedded Brownian motions are used to study estimator behavior and test power under deviations from the Markov assumption.
Additional simulations for varying values of $p$ are reported in \Cref{apx:numerics} for both estimation and testing.

\subsection{Covariance estimation}

We illustrate how the Markov approach influences the structure of the covariance estimator and compare its shape to both the empirical covariance and standard smoothing-based estimators in \Cref{fig:estimators_shape}. Specifically, we compare the following methods:
\begin{enumerate}[nosep]
    \item \textbf{The Markov estimator} \eqref{eq:markov-estimator};
    \item \textbf{The empirical covariance};
    \item \textbf{The smoothed covariance estimator}, obtained by smoothing the entire covariance kernel \citep{yao2005functional};
    \item \textbf{The triangular estimator}, obtained by smoothing over the upper triangle of the kernel and imposing symmetry \citep{mohammadi2024rough}.
\end{enumerate}
\begin{figure}[h!]
    \centering
    \includegraphics[width=0.8\textwidth]{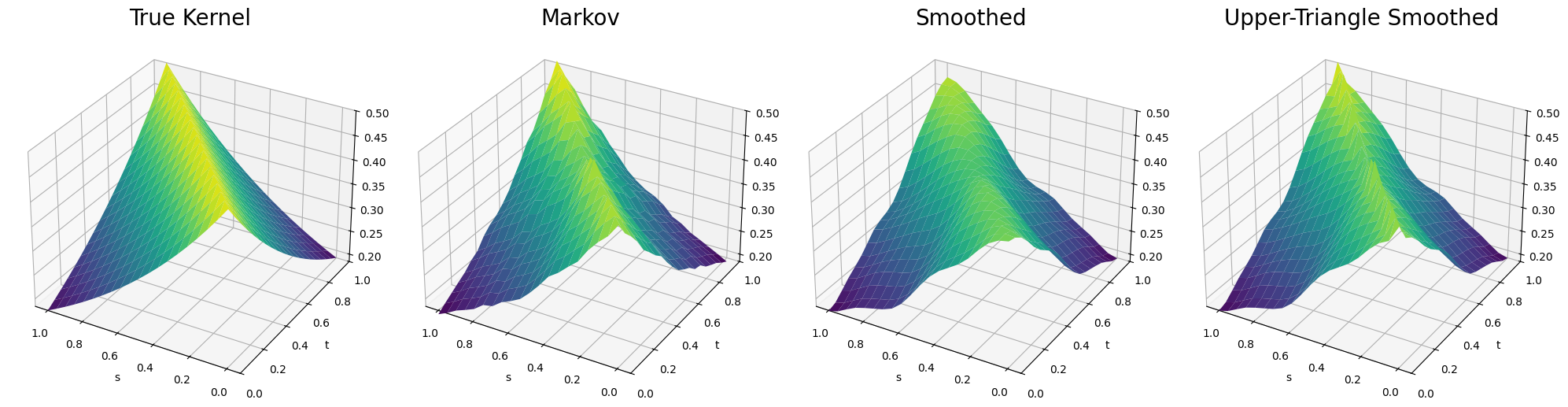}
    \caption{Comparison of empirical kernel shapes for the Ornstein-Uhlenbeck process ($p=20$, $n=200$).}
    \label{fig:estimators_shape}
\end{figure}

We first assess the performance of the estimator across both Markov and near-Markov processes using mean-squared error. 
As shown in \Cref{fig:L2-cv}, the Markov estimator performs comparably to the empirical covariance and outperforms the smoothed and triangular estimators when the underlying process belongs to the Markov class. 
For processes outside the Markov class, the estimator reaches a plateau in convergence, due to the bias incurred by the Markov assumption, with the residual distance to the true kernel determined by \eqref{eq:misspecification}.

\begin{figure}[h!]
    \centering
    \begin{minipage}{0.45\textwidth}
        \centering
        \includegraphics[width=\textwidth]{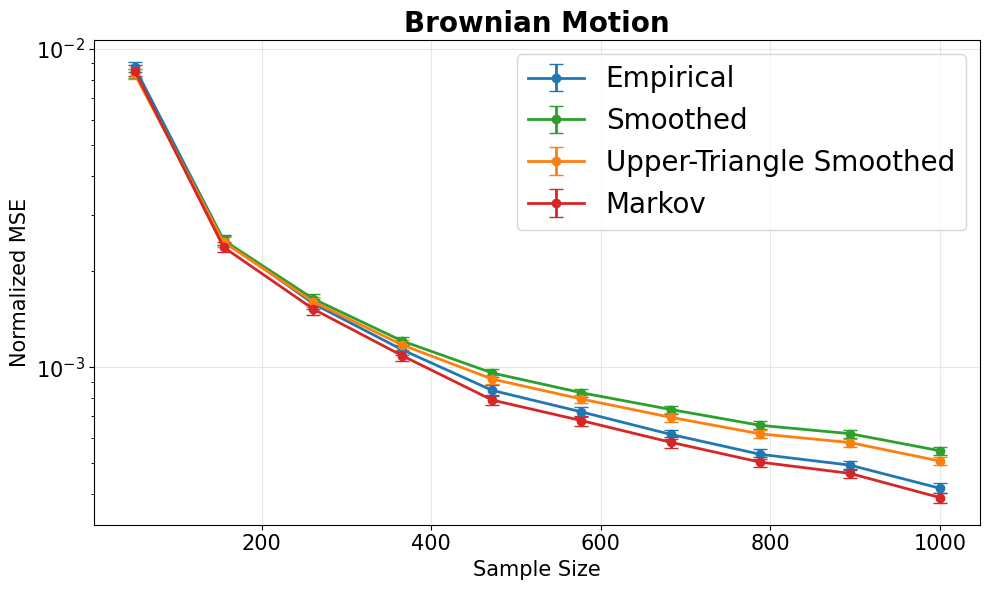}
    \end{minipage}
    \begin{minipage}{0.45\textwidth}
        \centering
    \includegraphics[width=\textwidth]{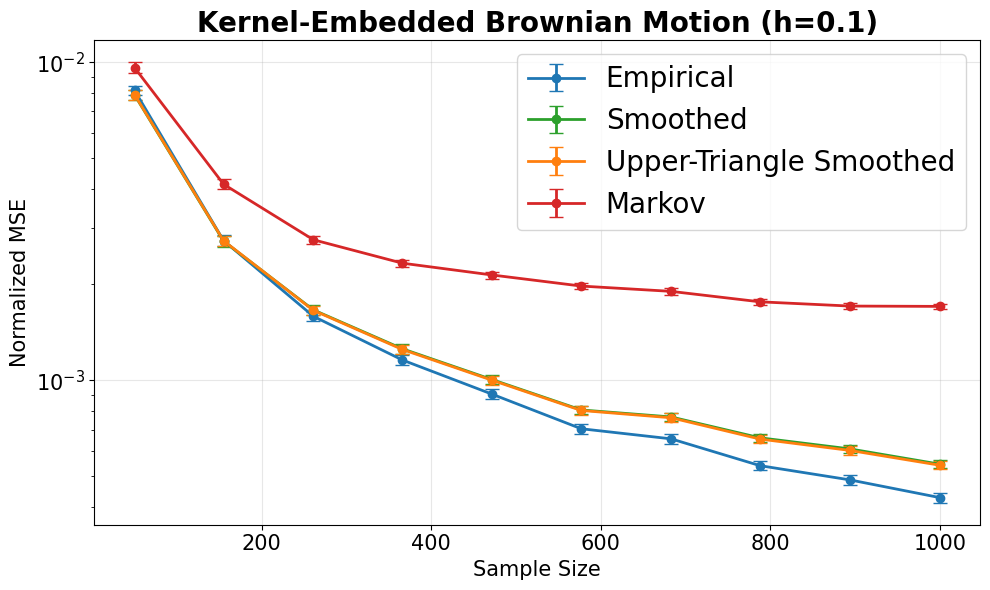}
    \end{minipage}

    \caption{Comparison of convergence in $L_2$-norm as $n$ grows for the Brownian motion (left) and the Kernel-Embedded Brownian motion with $h=0.1$ (right) for $p = 20$.}
    \label{fig:L2-cv}
\end{figure}

\subsection{Kriging and prediction}

A key advantage of the Markov approach is that it enforces sparsity at the level of the conditional-independence structure, or equivalently, the precision matrix in finite dimension. 
Accordingly, we expect the principal advantages of the Markov estimator to manifest in tasks where the inverse covariance plays a key role, such as kriging. Assume we observe a stochastic process $X$ on a compact interval $I$ at points $\{t_1, \dots, t_p\}$, and we wish to predict the value $X(t_0)$ at an unobserved location $t_0 \in I$. We focus on linear predictors of the form
$\widehat{X}(t_0) = \sum_{j=1}^{p} w_j X(t_j)$, with $w_j \in \RR$.
The classical kriging problem seeks weights that minimize the mean squared prediction error among all unbiased predictors:
\begin{align*}
    (w_1^*, \dots, w_p^*)^\top 
    &= \arg\min_{w \in \RR^p} 
    \mathrm{Var}\Big[X(t_0) - \sum_{j=1}^p w_j X(t_j)\Big] 
    \quad \text{s.t. } \mathbb{E}[X(t_0)] = \sum_{j=1}^p w_j \mathbb{E}[X(t_j)].
\end{align*}
For processes with stationary mean, this constraint reduces to $\sum_{j=1}^p w_j = 1$, and the optimal weights can be computed by solving the linear system
\begin{align*}
    \begin{bmatrix}
        K(\mathbf{t},\mathbf{t}) & \mathbf{1} \\
        \mathbf{1}^\top & 0
    \end{bmatrix}
    \begin{bmatrix}
        \mathbf{w}^* \\ \mu
    \end{bmatrix}
    =
    \begin{bmatrix}
        K(t_0,\mathbf{t}) \\ \mathbf{1}
    \end{bmatrix},
\end{align*}
where $K(\mathbf{t},\mathbf{t})$ is the covariance matrix of the observed points, 
$K(t_0,\mathbf{t}) = (K(t_0, t_j))_j$, 
and $\mu$ is a Lagrange multiplier enforcing the unbiasedness constraint.
This formulation makes explicit that the predictor depends on the inverse of the covariance matrix, which is central to assessing the performance of the Markov estimator. 
Before inversion, we regularize the covariance $K(\mathbf{t}, \mathbf{t})$ by adding a small ridge term $\varepsilon I_p$.
In our simulations, we study the distribution and mean squared error of the prediction error
$\epsilon(t_0) = \sum_{j=1}^{p} w_j X(t_j) - X(t_0)$,
highlighting how the structural constraints imposed by the Markov property impact kriging performance.
\begin{figure}[h!]
    \centering
    \includegraphics[width=0.9\textwidth]{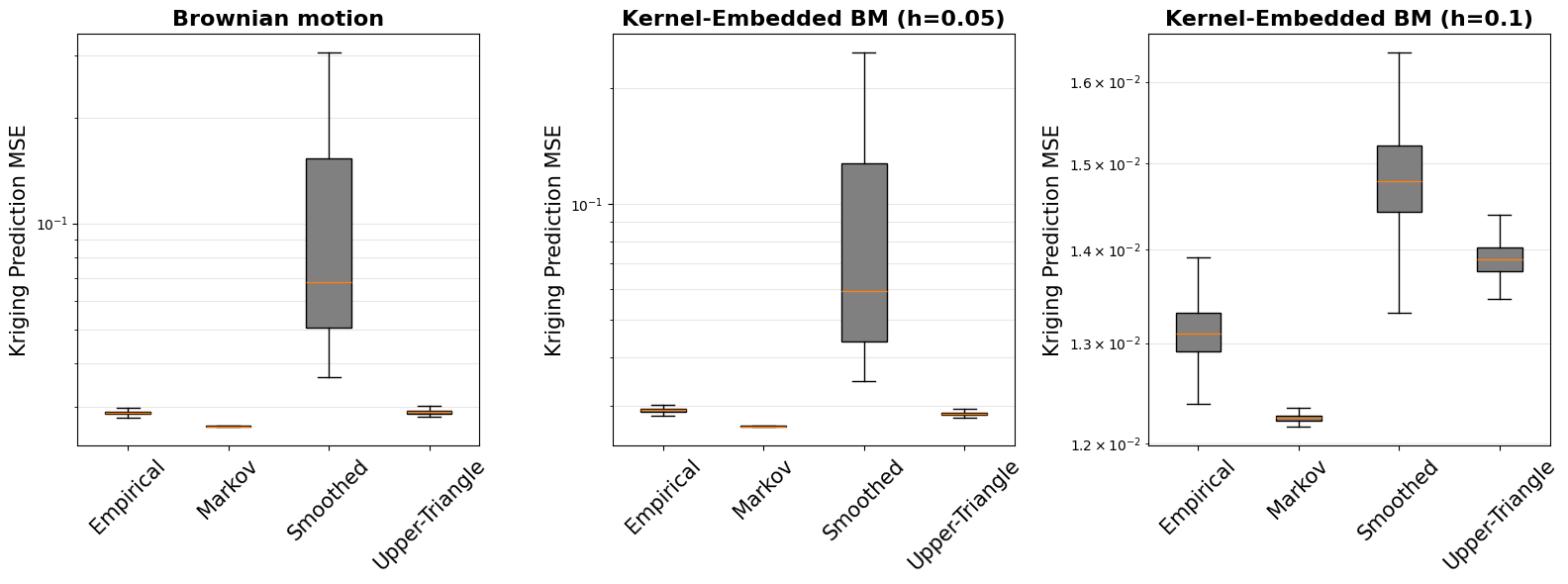}
    \caption{Kriging errors distributions for Brownian motion and kernel-embedded Brownian motion ($p=20$, $n=200$).}
    \label{fig:kriging-boxplot}
\end{figure}

We observe that, in terms of kriging, the Markov estimator outperforms both the empirical covariance estimator and smoothing-based methods, by an order of magnitude. 
Importantly, this remains true even when the underlying process is not strictly Markov: the Markov transform acts as an stabilizer for tasks that involve implicitly the inverse covariance matrix, leading to smaller prediction errors.

\begin{remark}
    To avoid committing an inverse crime \citep{kaipio2005statistical}, we use a finer grid to generate the data than to evaluate the estimator.
\end{remark}

\subsection{Testing power}

Next we turn to the empirical performance of the testing procedure, in terms of power. We study this by way of the embedded Brownian motion, as a function of the bandwidth parameter $h$. Embedding the Brownian motion covariance with a smoothing kernel introduces local dependencies that violate the Markov property. 
When the kernel bandwidth is smaller than the discretization step, the embedding has essentially no effect and the observed process remains virtually Markov. 
As the bandwidth $h$ increases, the deviations from the Markov structure become more pronounced. 
Hence, the test exhibits low power for very small values of $h$, while the rejection probability increases as $h$ grows. 
We use the maximum aggregation scheme, as it is particularly effective for detecting local deviations of conditional independence, which correspond to the considered alternatives.

\begin{figure}[h!]
\centering
\includegraphics[height=0.25\textheight]{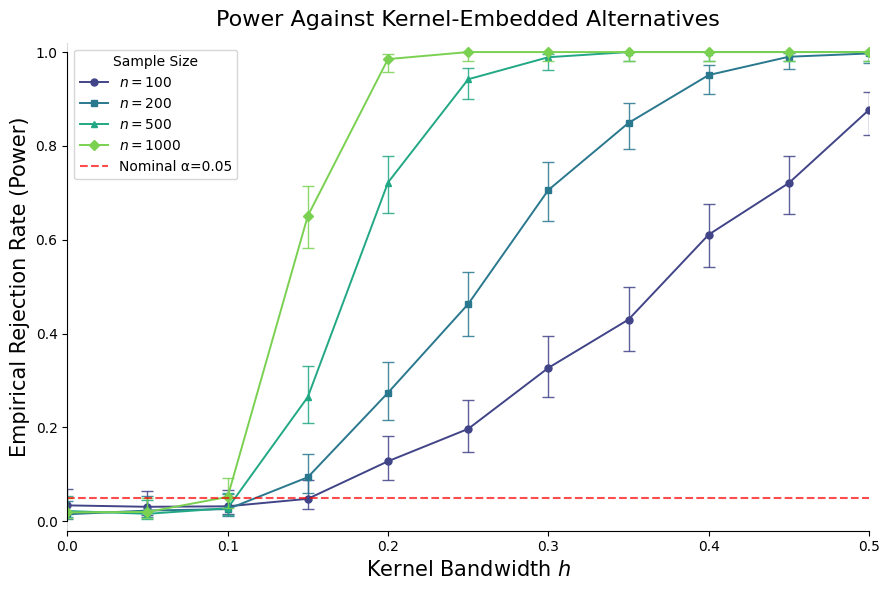}
\includegraphics[height=0.25\textheight]{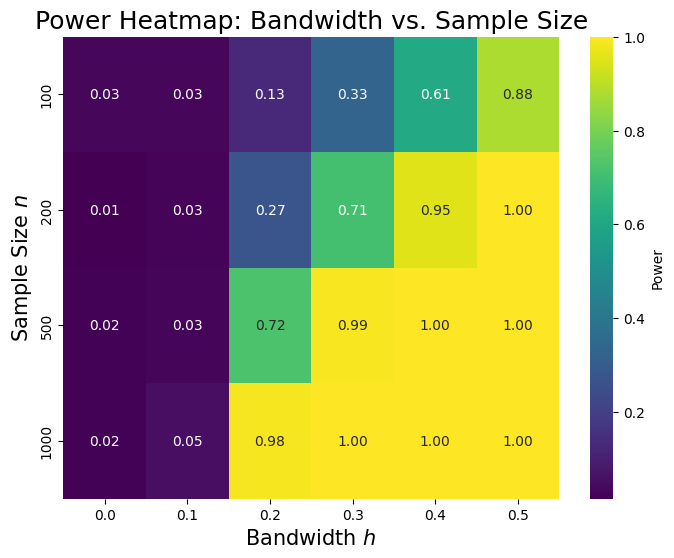}
\caption{Power of the proposed test against local alternatives generated by kernel-embedded Brownian motions with bandwidth $h \in [0,0.5]$, for $p=20$.}
\end{figure}

\section{Concluding remarks}\label{sec:discussion}

A key message of this paper is that smoothness is neither the only, nor a necessary, route to functional-data inference: structural parsimony can provides an alternative and, in some settings,  preferable foundation for statistical analysis. Concretely, we studied sparsity at the level of the conditional-independence structure of the observed functions and showed that the Markov property induces an interpretable shape constraint on the covariance kernel. Our results lead to a number of further questions, which we now discuss.

A natural extension of the Markov estimator is a shrinkage variant -- where  one would interpolate between the empirical covariance and its Markov projection (for example along a geodesic in the cone of positive-definite kernels). This would allow for a tunable bias and variance tradeoff, at the expense of introducing a tuning parameter. Such shrinkage can act as a parsimony-inducing regularization analogous in spirit to the LASSO and operates at the level of the inverse covariance. The expectation would be for such a strategy to retain the improved prediction performance of the Markov estimator, while allowing for consistency even under non-Markovian processes.

A different direction would be to allow for more general structural parsimony through conditional independence. For instance, one could replace the Markov property by a Markov property with memory $\delta>0$, corresponding to $\delta$-banded graphical constraints at the level of the continuous process. If the domain is partitioned into blocks at a scale larger than $\delta$, then the resulting covariance operator matrix exhibits a block Markov structure, and therefore admits an analogous operator-level characterization to the Markov transform. Exploring inference and projection operators for banded, block-structured, or other graph-constrained covariance classes seem natural in this context.

Finally, Markovianity can be seen as introducing \emph{a shape constraint} on the covariance kernel of a Gaussian process. This observation suggests a broader nonparametric perspective in which other constraints could be exploited. Covariance kernels often exhibit specific structural features, such as monotonicity or convexity relative to the diagonal, that can be incorporated directly through shape-constrained methods. Developing such approaches to covariance estimation could relax reliance on smoothness assumptions while remaining compatible with a wide class of stochastic processes typically encountered in functional data analysis.

\bibliography{references.bib}

\appendix

\input{appendix}

\end{document}

%% file: macros.tex
\usepackage{fullpage}

\usepackage{amsthm}
\usepackage{amsmath}
\usepackage{amsfonts}
\usepackage{amssymb}
\usepackage{euscript}
\usepackage{enumitem}
\usepackage{booktabs}
\usepackage{threeparttable}
\usepackage{multirow}
\usepackage{makecell}
\usepackage{authblk}

\usepackage[colorlinks, hyperfootnotes=false]{hyperref}
\hypersetup{
	linkcolor=blue,
	citecolor=blue,
}
\usepackage[english]{cleveref}
\usepackage{algorithm, caption}
\usepackage{algpseudocode}
\usepackage{mathtools}
\usepackage{multirow}
\usepackage{titling}
\usepackage{listings}
\usepackage{outlines}
\usepackage{graphicx}
\usepackage{subfig}
\usepackage{dsfont}
\usepackage[colorinlistoftodos,bordercolor=orange,backgroundcolor=orange!20,linecolor=orange,textsize=scriptsize]{todonotes}
\usepackage{setspace}
\usepackage{natbib}
\usepackage[justification=centering]{caption}
\usetikzlibrary{calc}
\usetikzlibrary{spy}

\theoremstyle{plain}
\newtheorem{theorem}{Theorem}[section]
\newtheorem{lemma}[theorem]{Lemma}
\newtheorem{proposition}[theorem]{Proposition}
\newtheorem{corollary}[theorem]{Corollary}
\newtheorem{remark}[theorem]{Remark}

\newtheorem{definition}[theorem]{Definition}

\newcommand{\NN}{\mathbb{N}}

\newcommand{\RR}{\mathbb{R}}

\newcommand{\PP}{\mathbb{P}}
\newcommand{\EE}{\mathbb{E}}

\newcommand{\cH}{\mathcal{H}}

\newcommand{\cF}{\mathcal{F}}
\newcommand{\cE}{\mathcal{E}}

\newcommand{\cM}{\mathcal{M}}

\newcommand{\cN}{\mathcal{N}}

\newcommand{\cI}{\mathcal{I}}
\newcommand{\cJ}{\mathcal{J}}

\newcommand{\bp}{\mathbf{p}}

\newcommand{\bR}{\mathbf{R}}
\newcommand{\bK}{\mathbf{K}}
\newcommand{\bP}{\mathbf{P}}

\newcommand{\bX}{\mathbf{X}}

\newcommand{\norm}[1]{\left\lVert#1\right\rVert}
\newcommand{\innerp}[2]{\langle #1, #2 \rangle}
\newcommand{\ind}{\mathrel{\perp\mspace{-9mu}\perp}}
\newcommand{\notind}{\mathrel{\,\not\!\ind}}

\DeclareMathOperator*{\argmin}{arg\,min}

\newcommand{\cov}{\operatorname{cov}}
\newcommand{\var}{\operatorname{var}}
\newcommand{\tr}{\operatorname{tr}}
\newcommand{\Tri}{\operatorname{Tri}}

\newcommand{\RhatM}{\widehat{R}^{\operatorname{M}}}
\newcommand{\KhatM}{\widehat{K}^{\operatorname{M}}}

\newcommand{\Klin}{K^{\operatorname{lin}}}
\newcommand{\Rlin}{R^{\operatorname{lin}}}

\newcommand{\OmX}{\Omega_X}

%% file: appendix.tex
\section*{Appendix}\label{apx}

\section{Graphical models}\label{apx:graphical-models}

Conditional dependencies and independencies characterize how a collection of random variables relate to one another, both directly and indirectly. 
Their structure can be encoded by a graph, which provides a concise representation of the interactions among the variables. 
Moreover, the conditional independence relation and its associated graph share key structural properties, making graphs particularly relevant for studying conditional dependence structures. 
We now formally define undirected graphs, which form the underlying structure of the graphical models.
An undirected graph is a pair $\mathcal{G} = (V, E)$ where $V$ is a set of vertices and 
$E \subset V \times V$ is a set of edges between the vertices with $(v, v) \notin E$ for all $v$.
We first consider graphs with a finite number of nodes, 
i.e. we have a finite number of variables, 
where $V = \{1, \dots, p\}$. Let $Z \in \RR^p$ and $S \subseteq [p]$.
We relate conditional independence and graphs by defining the conditional independence graph of a distribution.
We denote the subvector of $Z$ whose components are indexed by $S$ by $Z_S \in \RR^{|S|}$. 
Moreover, if $j, k \in [p]$, we write $Z_{-j} \coloneqq Z_{[p]\setminus\{j\}}$ and 
$Z_{-jk} \coloneqq Z_{[p]\setminus\{j, k\}}$.
\begin{definition}\label{def:cond-ind-graph}
    The conditional independence graph of a distribution $P$ on $\RR^p$ is the undirected graph $\mathcal{G} = (V, E)$ where given $Z \sim P$, 
    \begin{equation*}
        \{ j,k \} \in E \iff Z_j \notind Z_k \mid Z_{-jk}.
    \end{equation*}
\end{definition}
It results from this definition that the conditional independence graph of a 
finite sequence of random variables characterizes exactly its
conditional independence structure.
Undirected graphical models can also be of interest when
it comes to modeling functional data. Two perspectives 
exist to leverage conditional independence graphs
in the context of functional observations. The first consists in
considering each observation as a random element of a Hilbert space
and representing this element by a node in the graph. One can then
study how several functions are related in terms of conditional dependency.
This approach was studied in both dense and sparse observation settings
\citep{qiao2019functional, qiao2020doubly, waghmare2025functional}. 
The second approach consists in looking at each observed function as a random
process, i.e. a collection of random variables indexed by a continuous
set and was introduced in \citep{waghmare2025continuously}.
Note that this second perspective using continuously indexed graphs 
incorporates the first perspective as
a vector of functions can be represented as a single process by concatenating
each coordinate together. 
\begin{definition}
    Let $\Omega$ be an undirected graph on $I$, and $s,t \in I$ with $s \neq t$.
    We say that there exists a path between $s$ and $t$ in $\Omega$
    if there is a sequence of distinct points of $I$
    $s_1, s_2, \dots, s_n$ with $n \geq 2$ such that 
    $s_1 = s$, $s_n =t$ and
    $(s_i, s_{i+1}) \in \Omega$ for all $i \in \{1, \dots, n-1\}$.
    We further say that the set $S$ separates $A$ and $B$ in $\Omega$
    if every path from $A$ to $B$ in $\Omega$ intersects $S$.
\end{definition}
If $s$ and $t$ satisfy the above conditions,
we denote by $\bp(s,t)$ the sequence $(s_1, \dots, s_n)$.
Let $X = \{X_t \: : \: t\in I\}$ be a collection of random variables 
where $I$ is a continuous index set and denote $X_W = \{X_t \: : \: t\in W\}$ the
restriction of the process to the index set $W \subset I$. 
We can now define graphical models on continuous sets and
the graph of a stochastic process.
\begin{definition}
    We say that $(X, \Omega)$ is a graphical model if 
    $X_s \ind X_t \mid X_W$ for all $s, t$ separated by $W$ in $\Omega$.
    We call the graph of the process $X$ the intersection
    of all closed graphs $\Omega$ for which $(X, \Omega)$
    is a graphical model.
\end{definition}

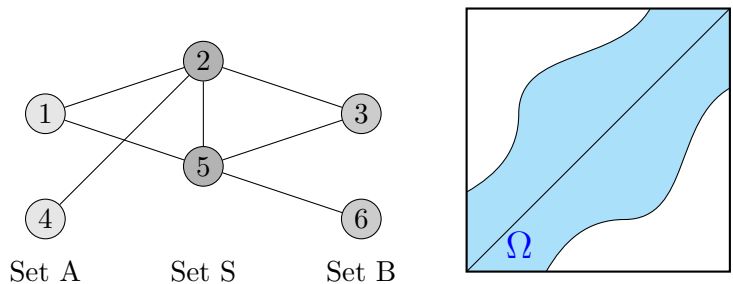
\begin{figure}[h!]
    \centering
    \begin{tikzpicture}[scale=0.7,
        every node/.style={circle, draw, inner sep=1pt, minimum size=15pt}
    ]
    \begin{scope}[shift={(-4,1)}]
        \node[circle, draw, fill=gray!20] (A1) at (0, 2) {1};
        \node[circle, draw, fill=gray!20] (A2) at (0, 0) {4};
        \node[circle, draw, fill=gray!60] (S1) at (3, 3) {2};
        \node[circle, draw, fill=gray!60] (S2) at (3, 1) {5};
        \node[circle, draw, fill=gray!40] (B1) at (6, 2) {3};
        \node[circle, draw, fill=gray!40] (B2) at (6, 0) {6};
        \draw (A1) -- (S1);
        \draw (A1) -- (S2);
        \draw (A2) -- (S1);
        \draw (S1) -- (B1);
        \draw (S2) -- (B1);
        \draw (S2) -- (B2);
        \draw (S1) -- (S2);
        \node[draw=none] at (0, -1) {Set A};
        \node[draw=none] at (3, -1) {Set S};
        \node[draw=none] at (6, -1) {Set B};
    \end{scope}

    \begin{scope}[shift={(4,0)}]
        \draw[thick] (0, 1.5) to[out=30, in=270] (1, 3) to[out=90, in=240] (3.5, 5);
        \draw[thick] (1.5, 0) to[out=60, in=180] (3, 1) to[out=0, in=-150] (5, 3.5);
        \fill[cyan!30] (0, 1.5) to[out=30, in=270] (1, 3) to[out=90, in=240] (3.5, 5) -- (5, 5) -- (5, 3.5) to[out=-150, in=0] (3, 1) to[out=180, in=60] (1.5, 0) -- (0, 0) -- cycle;
        \draw[thick] (0, 0) rectangle (5, 5);
        \draw (0, 0) -- (5, 5);
        \node[blue, font=\Large, draw=none] at (1, 0.5) {$\Omega$};
    \end{scope}
    \end{tikzpicture}
    \caption{Example of a discrete (left) and a continuous (right) graphical model.
    In the discrete graph, the set $S$ separates sets $A$ and $B$.}
    \label{fig:both}
\end{figure}

Note that if $\Omega$ is a
graphical model for $X$, then every graph $\Omega'$ 
containing $\Omega$ is also a graphical model for $X$.
Therefore, the graph of a process corresponds intuitively
to the smallest graph for which $(X, \Omega)$ is a graphical model.
It is possible to characterize graphical models of Gaussian processes in
terms of inner products of the covariance kernel in a reproducing kernel Hilbert space
(see \textit{Theorem 3.1} in \citep{waghmare2025continuously}). 
\begin{remark}\label{remark:notind-path}
    Let $X$ be a stochastic process defined on $I$, $\Omega$
    a graphical model for $X$, and suppose that
    $X_s \notind X_t \mid X_U$. Then there is a path in
    $\Omega$ not crossing $U$.
\end{remark}
It is important to observe that for continuously indexed sets
the graph of a process lacks two properties.
First, there is no guarantee that the graph of a process 
will be a graphical model for the process.
Second, if the graph $\Omega$ reflects some conditional independence
relations of $X$ when $(X, \Omega)$ is a graphical model
via properties of separation,
the converse is not true.
Knowing that $X_s \ind X_t \mid X_W$ does not imply that $s$ and $t$
are separated in $\Omega$. In fact, this is clear when
observing that the complete graph is a graphical model for every process.
Thus, it is not possible to characterize exactly the conditional independencies 
of a stochastic process by looking at its graph.
However, we prove that under some assumptions, it is possible
to restore this characterization by going through a finite
approximation of the graph (\Cref{thm:non-singular-correlation}).

\section{A kernel perspective on the endpoints characterization}
Let $X$ be a real-valued Gaussian process defined on a compact interval $I$ with continuous sample paths and nowhere degenerate.
We stated (\Cref{thm:endpoints-cond-ind}) that, if $X$ has a non-singular correlation structure and is strongly locally nondeterministic on $[a,b]$, where $a,b\in I$ with $a<b$, then
$X$ is Markov on the interval $[a, b]$ if and only if 
\begin{equation*}
        X_a \ind X_b \mid X_c \quad \text{for all } c\in (a,b).
\end{equation*}

It follows from this result that the conditional independence structure of a Gaussian Markov process possesses only one degree of freedom.
This has implications at the level of the covariance kernel of the process, and enables us to revisit the characterization of Gauss-Markov kernels given in \Cref{sec:markov-covariance-model}.
We first propose a reformulation of the functional characterization of the covariance kernel for Markov processes using endpoints.

\begin{proposition}\label{prop:cov-endpoints}
    Let $X$ be a Gaussian stochastic process defined on a compact interval $I$ 
    with almost-surely continuous sample paths.
    Let $a, b \in I$ such that $a < b$ and 
    assume that the process is not degenerate on $[a,b]$.
    Then $X$ is a Markov process on $[a,b]$ if and only if its covariance kernel $K$ satisfies
    \begin{equation}
        K(s,t) = \frac{1}{K(a,b)}K(a,t)K(s,b) \quad \text{for all } a\leq s\leq t\leq b.
    \end{equation}
\end{proposition}

This implies that knowing the whole covariance is equivalent to knowing the covariance between each of the endpoints and all the points in-between.
We give a reformulation of \Cref{thm:endpoints-cond-ind} at the level
of covariance kernels.

\begin{corollary}\label{thm:kernel-endpoints}
    Under the assumptions of \Cref{thm:endpoints-cond-ind}, 
    $X$ is a Markov process on $[a,b]$ if and only if its covariance kernel $K$ satisfies
    \begin{equation}\label{eq:kernel-endpoints}
        K(a,b) = \frac{K(a,u)K(u,b)}{K(u,u)} \quad \text{for all } u \in [a,b].
    \end{equation}
\end{corollary}

This formalizes the idea that the covariance kernel of a Gaussian Markov process is in essence a one-dimensional function.
\begin{remark}
    At the level of rescaled kernels, \Cref{eq:kernel-endpoints} translates as
    $R(a,b) = R(a,u)\cdot R(u,b)$ for all $u \in [a,b]$.
\end{remark}
This suggests another perspective on tests for the Markov property,
that is, assessing whether \Cref{eq:kernel-endpoints} holds.
Let $v < w$ be two distinct times where the process $X$ is observed.
Recall that since $X$ is Gaussian,
we can write
$X_w = X_v \beta_{v,w} + \epsilon_{v,w}$,
where $\epsilon_{v,w}$ follows a normal distribution.
The regression coefficient is given by:
$\beta_{v,w} = R(v,w)$.
Therefore, we can use the empirical regression coefficient $\hat{\beta}_{v,w}$
obtained by regressing the observed values of the process at $w$ onto the ones
at $v$ to estimate $R(v,w)$.
Thus, to test whether the considered process satisfies
the Markov property, we regress $X_u$ on $X_a$ and $X_b$ on $X_u$
for all observation times $u \in [a,b]$, and compute a test statistic
from the regression coefficients. 
Note that this is essentially equivalent to the approach proposed in 
\Cref{sec:test-markov} as standard conditional independence
testing approaches consist in performing the exact same regressions.

\section{Proofs and technical details}

\subsection{The Markov transform}

\begin{proof}[Proof of \Cref{prop:markov-projection}]
$(1) \iff (3)$:
Cauchy-Schwartz inequality gives us that $\sigma^2_{t_{j+1}} \geq 0$.
Hence the process is well-defined and a direct computation shows that its covariance is $C^M$.

$(1) \iff (2)$: For simplicity, write $C \coloneqq C_p$
and $\Omega=\Sigma^{-1}$. 
Denote 
\begin{align*}
    f:\Omega \in \Tri_p \mapsto KL(\cN(0, C), \cN(0, \Sigma)) = \tr(\Omega C) - \log \det \Omega,
\end{align*}
where $\Tri_p$ is the space of tridiagonal symmetric PD matrices.
By convexity of $\Tri_p$ and strict convexity of $f$, there exists a unique minimizer of $f$.
Let $\Sigma^\star=\Omega^{\star-1}$.
Since $\Omega^*$ is tridiagonal, the vector 
$(X^*_1, \dots, X^*_p) \sim \cN(0, \Sigma^*)$ has a Gaussian Markov distribution, and its covariance is given by
\begin{align*}
    \Sigma^*_{ij} = \var(X^*_i) \cdot  \left(\prod_{k=i}^{j-1}\beta^*_k\right) \quad \text{where }
    \beta^*_k = \frac{\cov(X^*_k, X^*_{k+1})}{\var(X^*_k)},
\end{align*}
for $i < j$.
Moreover, observe that by first-order optimality conditions:
$\Sigma^{\star}_{ij}=C_{ij}$ for all $|i-j|\leq 1$.
It follows that $\cov(X^*_k, X^*_{k+1}) = \cov(X_k, X_{k+1})$ and $\var(X^*_k) = \var(X_k)$ for all $k \in [p]$.
This concludes the proof.
\end{proof}

\subsection{Rates of convergence}

\paragraph{Proof of the main result.}

The discrete rates are obtained thanks to two key observations.
First, the regression coefficients used in the estimator are uncorrelated due to the Markov property.
Second, the continuity enables the control of these regression coefficients as $p$ grows.

\begin{lemma}\label{lemma:asymptotic-independence}
    Let $X$ be a real-valued, mean-zero Gaussian process on the interval $I$.
    If the process $X$ is Markov, the random variables $\{\sqrt{n}(\hat{\beta}_j - \beta_j)\}_{1 \leq j \leq p}$ are uncorrelated for all $n \geq 1$. 
\end{lemma}

\begin{lemma}\label{lemma:sparse-uniform-continuity}
Let $X$ be a Gaussian process on the compact interval $I$ with mean zero and continuous covariance kernel $K$.
Let $w$ be the modulus of continuity of $K$.
Then, there are constants $C_1, C_2, C_3 > 0$ such that for all $k \in [p]$:
\begin{align*}
    |\beta_k - 1| \leq C_1 w\left(\frac{2}{p}\right)
    \quad \text{and} \quad 
    \var(\hat{\beta}_k) \leq \frac{C_2}{n_k-2} \left(w\left(\frac{2}{p} \right) + C_3 \nu^2 \right).
\end{align*}
\end{lemma}

Furthermore, we describe the asymptotic behavior of the number of observations in the irregular observation settings.
\begin{lemma}\label{lemma:sparse-growth}
    The random variables $n_k$, $k \in [r]$ are identically distributed
    and follow a binomial distribution
    with parameters $n$ and $\widetilde{p} = 1 - 2\left( 1 - \frac{1}{p} \right)^r + \left( 1 - \frac{2}{p} \right)^r$.
    When $p$ grows,
    \begin{align*}
         \EE\left[\frac{1}{n_k+1}\right] =
         \begin{cases}
             O\left( \frac{p^2}{nr^2}  \right) &\text{ if } r= o(p), \\
             O\left( \frac{1}{n}  \right) &\text{ otherwise}.
         \end{cases}
    \end{align*}
\end{lemma}

\begin{proof}[Proof of \Cref{thm:continuous-rates}]
Define
\begin{align*}
        R: (s,t) \in \Delta \mapsto \frac{K(s,t)}{K(s,s)} \quad \text{where }\Delta = \{(s,t) \in I \times I \; : s \leq t\}.
\end{align*}
Fix $p \in \NN$ and let $k \in [p-1]$.
By the mean-value theorem and as the covariance $K$ is continuous, there is $t^*_k \in B_k$ and $t^*_{k+1} \in B_{k+1}$ such that $\hat{\beta}_{k+1}$ converges to $R(t^*_k, t^*_{k+1})$.
Therefore, there is a sequence $(t^*_1, \dots, t^*_p)$ such that $\KhatM_p$ consistently estimates $K_p \coloneqq (K(t^*_i, t^*_j))_{i,j \in [p]}$ at a rate of order $1 / \sqrt{n_k}$.
Let $\Klin$ be the piecewise linear surface interpolating $(K(t^*_i, t^*_j))_{i,j \in [p]}$ on $[0,1] \times [0,1]$.
By the triangle inequality,
\begin{align*}
    \left(\EE \left[ \norm{\KhatM - K}_2^2 \right]\right)^{1/2} 
    \leq \left(\EE \left[ \norm{\KhatM - \Klin}_2^2 \right]\right)^{1/2} + \norm{\Klin - K}_2.
\end{align*}
The first term can be upper bounded by the mean-squared error of the matrix $\KhatM_p - K_p$ times $2/p$ corresponding to the grid size. 
Moreover, as $K$ has modulus of continuity $w$, the second term can be bounded as
$\norm{\Klin - K}_2 \leq w(h)$ where $h$ is the maximum distance between two observation points $t_{ij}$, $t_{ij'}$ that belong in consecutive bins (formally $(j,j') \in \cJ_i(k) \times \cJ_i(k+1)$ with $ k \in [p-1]$ and $i \in [n]$): in particular, $h \leq 2/p$.
Moreover, the entrywise product function is Lipschitz on any bounded domain.
As the covariance $K$ is bounded on $I$, there exists $L > 0$ such that
\begin{align*}
    \left(\EE \left[ \norm{\KhatM - \Klin}_2^2 \right]\right)^{1/2} \leq L \left\{ \left(\EE \left[ \norm{\RhatM - \Rlin}_2^2 \right]\right)^{1/2} + \left(\EE \left[ \norm{\hat{\sigma}^2 - \sigma^2}_2^2 \right]\right)^{1/2} \right\}
\end{align*}
Hence,
\begin{align}\label{apx:triangle-inequality}
    \left(\EE \left[ \norm{\KhatM - K}_2^2 \right]\right)^{1/2} 
    \leq \frac{2L}{p} \left(\EE \left[ \norm{\RhatM_p - R_p}_F^2 \right]\right)^{1/2} + \frac{L}{\sqrt{n}} + w\left(\frac{2}{p}\right).
\end{align}

We focus on the first term of the right-hand side.
Recall that
\begin{align*}
    \RhatM_p(t^*_j, t^*_k) = \prod_{l=1}^{k-j} \hat{\beta}_{j+l} \quad \text{for } j \leq k.
\end{align*}
Let $j < k-1$.
For brevity, write $\RhatM_p(j, k) \coloneqq \RhatM_p(t^*_j, t^*_k)$ and $R_p(j, k) \coloneqq R_p(t^*_j, t^*_k)$.
We have:
\begin{align*}
    \RhatM(j, k) - R(j,k) &= \RhatM(j, k-1) \hat{\beta}_{k} - {R}(j, k-1) {\beta}_{k}\\
    &= \RhatM(j, k-1) \Big[ \hat{\beta}_{k} - \beta_{k} \Big] + \Big[ \RhatM(j, k-1) - R(j,k-1) \Big] \beta_{k} 
\end{align*}
Denote $\varepsilon_{j,k} \coloneqq \RhatM(j, k) - R(j,k)$ for all $k > j+1$, and 
$\xi_k = \hat{\beta}_k - \beta_k$.
We just showed that
\begin{align*}
    \varepsilon_{j,k} &= \Big( \varepsilon_{j,k-1} + R(j,k-1)\Big) \xi_k  + \varepsilon_{j,k-1} \beta_k \\
    &= \varepsilon_{j,k-1} \Big(\beta_k + \xi_k\Big) + R(j,k-1)\xi_k.
\end{align*}

Using the fact that $\varepsilon_{j,j}=0$, we obtain:
\begin{equation}\label{eq:proof-estimation-error}
    \varepsilon_{j,j+k} = \sum_{l=1}^{k} R(j,j+l-1)\xi_{j+l}
    \prod_{r=l+1}^{k} (\beta_{j+r} + \xi_{j+r}).
\end{equation}

By \Cref{lemma:asymptotic-independence}, $(\xi_{j+l}, \dots, \xi_k)$ are mean-zero and uncorrelated. Hence, when squaring the error and taking the expectation, the cross terms vanish and 
\begin{align*}
    \EE[\varepsilon_{j,j+k}^2] 
    &= \sum_{l=1}^{k}R(j,j+l-1)^2  \left( \EE [ \xi_{j+l}^2] 
    \prod_{r=l+1}^{k} \EE \left[ (\beta_{j+r} + \xi_{j+r})^2 \right]  \right) \\
    &= \sum_{l=1}^{k}R(j,j+l-1)^2 \var(\hat{\beta}_{j+l}) 
    \prod_{r=l+1}^{k} \left( \beta_{j+r}^2 + \var(\hat{\beta}_{j+r})  \right) \\
    &= \sum_{l=1}^{k}R(j,j+l-1)^2 \var(\hat{\beta}_{j+l}) 
    \prod_{r=l+1}^{k} \left( \beta_{j+r}^2 \right)
    \prod_{r=l+1}^{k} \left( 1 + \frac{\var(\hat{\beta}_{j+r})}{\beta_{j+r}^2}  \right).
\end{align*}
Observe that by Cauchy-Schwartz, $K(j,k) \leq \sqrt{K(j,j) K(k,k)}$ which implies 
$\prod_{l=j+1}^{k} \beta_l \leq \sigma_k / \sigma_j$.
As the covariance is non-null, continuous on the compact interval $I$, there exists a constant $C_0$ that upper-bound this product independently of $p$.
For the following, condition on $n_k$.
By \Cref{lemma:sparse-uniform-continuity}, 
$|\beta_k - 1| \leq C_1 w\left(\frac{2}{p}\right)$ and 
$\var(\hat{\beta}_k) \leq \frac{C_2}{n_k} \left(w\left(\frac{2}{p} \right) + C_3 \nu^2 \right)$ for some constants $C_1, C_2, C_3 > 0$.
Hence,
\begin{align*}
    \EE[\varepsilon_{j,j+k}^2] 
    &= C_0^2 \frac{C_2}{n_k} \left(w\left(\frac{2}{p} \right) + C_3 \nu^2 \right) \norm{R}^2_{\infty} \cdot
    \sum_{l=1}^{k} 
    \prod_{r=l+1}^{k} \left( 1 + \frac{\var(\hat{\beta}_{j+r})}{\beta_{j+r}^2}  \right).
\end{align*}
Since the covariance is non-null and continuous, there is a constant $C_4>0$ such that $1 / \beta_{j+r}^2 \leq C_4$ for all $r$. 
$\sum_{l=1}^{k} \prod_{r=l+1}^{k} \left( 1 + {\var(\hat{\beta}_{j+r})}/{\beta_{j+r}^2}  \right) 
\leq \sum_{l=1}^{k} \left( 1 + C_4/n \right)^{k-l} \leq k e^{C_4 k/n}$.
As $p/n$ is bounded, we have $C_5>0$ such that
\begin{align*}
    \EE[\varepsilon_{j,j+k}^2] 
    &=  C_0^2 C_2 C_5 \norm{R}^2_{\infty} \frac{k}{n_k} \left(w\left(\frac{2}{p} \right) + C_3 \nu^2 \right).
\end{align*}
By the tower property and \Cref{lemma:sparse-growth},
\begin{align}\label{eq:proof-pointwise-bound}
    \EE[\varepsilon_{j,j+k}^2] 
    &=  M  \left( \frac{k}{n} \vee \frac{kp^2}{nr^2} \right) \cdot \left(w\left(\frac{2}{p} \right) + C_3 \nu^2 \right),
\end{align}
where $M \coloneqq 2C_0^2 C_2 C_5 \norm{R}^2_{\infty}$ is independent of $n$ and $p$.
Thus, the mean squared error satisfies
\begin{align*}
    \EE \left[ \norm{\KhatM_p - K_p}_F^2 \right]
    &\leq M  \left( \frac{1}{n} \vee \frac{p^2}{nr^2} \right) \cdot \left(w\left(\frac{2}{p} \right) + C_3 \nu^2 \right)  \cdot \sum_{j=1}^{p} \sum_{k=1}^{p} k \\
    &\leq M p^2  \left( \frac{p}{n} \vee \frac{p^3}{nr^2} \right) \cdot \left(w\left(\frac{2}{p} \right) + C_3 \nu^2 \right). 
\end{align*}
\end{proof}

\paragraph{Proof of technical lemmas}

\begin{proof}[Proof of \Cref{lemma:asymptotic-independence}]
    Suppose $p = 2$, the case where $p \geq 2$ can be handled similarly.
    We know that:
    \begin{align*}
        \hat{\beta}_1 = \frac{\sum_{i=1}^{n} X_1^{(i)} X_2^{(i)}}{\sum_{i=1}^{n} \left(X_1^{(i)}\right)^2},
        \qquad
         \hat{\beta}_2 = \frac{\sum_{i=1}^{n} X_2^{(i)} X_3^{(i)}}{\sum_{i=1}^{n} \left(X_2^{(i)}\right)^2}.
    \end{align*}
    Let $\bX_2 = (X_2^{(1)}, \dots, X_2^{(n)})^\top$. By the tower property and since $X_1 \ind X_3 \mid X_2$:
    \begin{align*}
        \EE \left[ \hat{\beta}_1 \hat{\beta}_2 \right] 
        &= \EE \left[ \frac{\sum_{i=1}^{n} \sum_{j=1}^{n} X_1^{(i)} X_2^{(i)}  X_2^{(j)} X_3^{(j)}}{\sum_{i=1}^{n} \left(X_1^{(i)}\right)^2 \cdot \sum_{j=1}^{n} \left(X_2^{(j)}\right)^2} \right] \\
        &= \EE \left[ \frac{1}{\sum_{i=1}^{n} \left(X_2^{(i)}\right)^2 } \cdot
        \sum_{i=1}^{n} \sum_{j=1}^{n} \left(X_2^{(i)} X_2^{(j)}\right)  \EE\left[ \left. \frac{X_1^{(i)} }{\sum_{i=1}^{n} \left(X_1^{(i)}\right)^2}X_3^{(j)} \right| \bX_2 \right] \right] \\
        &= \EE \left[ \frac{1}{\sum_{i=1}^{n} \left(X_2^{(i)}\right)^2 } \cdot
        \sum_{i=1}^{n} \sum_{j=1}^{n} \left(X_2^{(i)} X_2^{(j)}\right)  \EE\left[ \left. \frac{X_1^{(i)} }{\sum_{i=1}^{n} \left(X_1^{(i)}\right)^2} \right| \bX_2 \right] \EE\left[ \left. X_3^{(j)} \right| \bX_2 \right] \right] \\
        &= \beta_2 \EE \left[ \frac{1}{\sum_{i=1}^{n} \left(X_2^{(i)}\right)^2 } \cdot
        \sum_{i=1}^{n} \sum_{j=1}^{n} X_2^{(i)} \left(X_2^{(j)}\right)^2  \EE\left[ \left. \frac{X_1^{(i)} }{\sum_{i=1}^{n} \left(X_1^{(i)}\right)^2} \right| \bX_2 \right] \right] \\ 
        &= \beta_2 \EE \left[ 
        \sum_{i=1}^{n}  X_2^{(i)}  \EE\left[ \left. \frac{X_1^{(i)} }{\sum_{i=1}^{n} \left(X_1^{(i)}\right)^2} \right| \bX_2 \right] \right] \\ 
        &= \beta_2 \EE \left[ 
        \frac{\sum_{i=1}^{n}  X_2^{(i)} X_1^{(i)} }{\sum_{i=1}^{n} \left(X_1^{(i)}\right)^2} \right] \\ 
        &= \beta_1 \beta_2.
    \end{align*}

\end{proof}

\begin{proof}[Proof of \Cref{lemma:sparse-uniform-continuity}]
    Fix $p \in \NN$ and let $k \in [p-1]$.
    By the mean-value theorem and as the covariance $K$ is continuous, there is $t^*_k \in B_k$ and $t^*_{k+1} \in B_{k+1}$ such that $\hat{\beta}_{k+1}$ converges to $R(t^*_k, t^*_{k+1})$.
    Therefore, there is a sequence $(t^*_1, \dots, t^*_p)$ such that $\KhatM_p$ consistently estimates $K_p \coloneqq (K(t^*_i, t^*_j))_{i,j \in [p]}$.
    Let $\beta_{k+1} = K(t^*_{k}, t^*_{k+1})/K(t^*_{k}, t^*_{k})$.
    
    \textbf{First inequality.}
    We have:
    \begin{align*}
        |\beta_{k+1} - 1| &= \left| \frac{K(t^*_k, t^*_{k+1})}{K(t^*_k, t^*_{k})} - 1 \right| \\
        &= \left|\frac{K(t^*_k, t^*_{k+1}) - K(t^*_k, t^*_{k})}{K(t^*_k, t^*_{k})}\right| \\
        &\leq \frac{w(t^*_{k+1} - t^*_k)}{K(t^*_k, t^*_k)}.
    \end{align*}
    As $t^*_k \in B_k$, $t^*_{k+1} \in B_{k+1}$, we have $|t^*_{k+1} - t^*_k| \leq 2/p$.
    As the covariance $K$ is continuous and non-null, there is $C_1>0$
    such that $1/K(t^*_k, t^*_k) \leq C_1$ for all $k\in[p]$ and for all $p$. 
    
    \textbf{Second inequality.}
    We have
    \begin{align*}
        \var(\hat{\beta}_{k}) 
        &= \frac{\widetilde{\sigma}_{k}^2}{n_k-2},
    \end{align*}
    where $\widetilde{\sigma}_{k}^2$ is the variance of the residual
    in the regression setting with noisy observations, that is
    \begin{align*}
        \widetilde{\sigma}_{k}^2 &= \var \left[ \left(\frac{1}{m_i(k)} \sum_{j \in \cJ_i(k)} u_{ij} \right) \beta_{k} - \left(\frac{1}{m_i(k+1)} \sum_{j \in \cJ_i(k+1)} u_{ij} \right) + \frac{1}{m_i(k) m_i(k+1)} \sum_{\substack{j \in \cJ_i(k) \\ j' \in \cJ_i(k+1)}} \epsilon_{ij, ij'} \right]
    \end{align*}
    where $\epsilon_{ij, ij'}$ is the regression error when regressing $X(t_{ij'})$ on $X(t_{ij})$, with $\EE \epsilon_{ij, ij'} = 0$ and 
    $\EE \epsilon_{ij, ij'}^2 \eqqcolon \sigma^2_{ij, ij'}$.
    By independence of the $u_{ij}$'s, and by Cauchy-Schwartz inequality,
    \begin{align*}
        \widetilde{\sigma}_{k}^2 &= \nu^2 (\beta_k^2 + 1) + 
        \frac{1}{m_i(k)^2 m_i(k+1)^2} \var \left[ \sum_{\substack{(j,j') \in \cJ_i(k) \times \cJ_i(k+1)}} \epsilon_{ij, ij'} \right] \\
        &\leq \nu^2 (\beta_k^2 + 1) + \frac{1}{m_i(k) m_i(k+1)} 
        \left(\sum_{\substack{(j,j') \in \cJ_i(k) \times \cJ_i(k+1)}} \var[\epsilon_{ij, ij'}] \right). 
    \end{align*}
    Moreover,
    \begin{align*}
        \var[\epsilon_{ij, ij'}] 
        &= K(t_{ij'}, t_{ij'}) - \frac{K(t_{ij}, t_{ij'})^2}{K(t_{ij}, t_{ij})} \\
        &= \frac{1}{K(t_{ij}, t_{ij})} \big[ K(t_{ij}, t_{ij}) K(t_{ij'}, t_{ij'}) - K(t_{ij}, t_{ij}) K(t_{ij}, t_{ij'}) \\
        &\quad \ + K(t_{ij}, t_{ij'})K(t_{ij}, t_{ij}) - K(t_{ij}, t_{ij'})^2 \big] \\
        &=  \left[ K(t_{ij'}, t_{ij'}) - K(t_{ij}, t_{ij'}) \right] + \frac{K(t_{ij}, t_{ij'})}{K(t_{ij}, t_{ij})} \left[ K(t_{ij}, t_{ij}) - K(t_{ij}, t_{ij'})\right] \\
        &\leq w(|t_{ij} - t_{ij'}|) \left[ 1 + \frac{K(t_{ij}, t_{ij'})}{K(t_{ij}, t_{ij})} \right]
    \end{align*}
    Hence, since $K$ is uniformly continuous, non-null, and bounded on $I$,
    \begin{align*}
        \widetilde{\sigma}_{k}^2 
        &\leq C \nu^2 + C' w\left( \max_{(j,j') \in \cJ_i(k) \times \cJ_i(k+1)}|t_{ij} - t_{ij'}|\right) \\
        & \leq C \nu^2 + C' w\left( \frac{2}{p} \right).
    \end{align*}   
\end{proof}

\begin{proof}[Proof of \Cref{lemma:sparse-growth}]
    We have $n_k = \sum_{i=1}^{n} \mathbf{1}(A_i)$ with 
    $A_i = \left\{ \exists (j, j') \text{ s.t. } j \in \cJ_i(k), j' \in \cJ_i(k+1) \right\}$, and,
    \begin{align*}
        \PP(A_i) &= 1 - 2\left( 1 - \frac{1}{p} \right)^r + \left( 1 - \frac{2}{p} \right)^r = \widetilde{p}.
    \end{align*}
    Noting that $(\mathbf{1}(A_i))_{i \in [n]}$ are i.i.d. concludes the proof of the first statement.
    By \citep{chao1972negative}, we have:
    \begin{align}
        \EE\left[\frac{1}{n_k+1}\right] = \frac{1 - (1-\widetilde{p})^{n+1}}{(n+1)\widetilde{p}},
    \end{align}
    and a Taylor expansion gives the rate.
\end{proof}

\subsection{Testing the Markov property}

\paragraph{Proof of the endpoints characterization of the Markov property}

\begin{theorem}\label{thm:non-singular-correlation}
    Let $X$ be a real-valued Gaussian process on a compact interval $I$ and let $\OmX$ be its graph.
    If $X$ has a non-singular correlation structure,
    then, for all $s,t \in I$ and closed set $U$ such that
    $X_s \ind X_t \mid X_U$, the points $s$ and $t$ are separated by $U$ in $\OmX$.
\end{theorem}
By definition of uncountable graphical models, the converse always holds.

\begin{proof}[Proof of \Cref{thm:non-singular-correlation}]
    Let $s,t \in I$ and $U$ be a closed set,
    and assume that $X_s \ind X_t \mid X_U$.
    If $s,t$ and $U$ are not strictly separated, the property clearly holds.
    Hence, we suppose that $s,t$ and $U$ are strictly separated.
    As $X$ has a non-singular correlation structure, we have a sequence $\{\pi_j\}_{j=1}^\infty$ such that $\bR_{\pi_j}$ is invertible for all $j \in \NN^*$.
    Let $j^* > 0$ such that $s, t$ and $U$ are separated in $\pi_j$ for all $j \geq j^*$.
    Then, for $j \geq j^*$, we have $X_s \ind X_t \mid X_{U^{\pi_j}}$.
    By definition of finite graphical models, and since $\bR_j$ is invertible,
    it follows that $U^{\pi_j}$
    separates $s$ and $t$ in $\Omega^\pi_j$.

    Let $\bp(s,t)$ we a path from $s$ to $t$ in $\cap_{j \geq j^*}\OmX^{\pi_j}$, 
    and fix $k \geq j^*$. Since $\cap_{j \geq j^*}\OmX^{\pi_j} \subset \OmX^{\pi_k}$,
    $\bp(s,t)$ is also a path from $s$ to $t$ in $\OmX^{\pi_k}$, and hence, it crosses
    $U^{\pi_k}$. As this is true for all $k \geq j^*$, it follows that 
    $\cap_{j \geq j^*} U^{\pi_j}$ separates $s$ and $t$ in $\cap_{j \geq j^*}\OmX^{\pi_j}$.

    Now, observe that for any set $W$, $\cap_{j \geq j^*} W^{\pi_j} = \overline{W}$ where $\overline{W}$
    is the closure of $W$. As $U$ is closed, it follows that
    $U$ separates $s$ and $t$ in $\overline{\OmX}$, and as $\OmX \subset \overline{\OmX}$ we obtain
    that $U$ separates $s$ and $t$ in $\OmX$.
\end{proof}

\begin{lemma}\label{lemma:slnd}
    Let $X$ be a strongly locally nondeterministic Gaussian process defined on $[0,1]$ and let $U \subset [0,1]$ be a closed interval.
    If $X$ is continuous, then there exists $\delta > 0$ such that
    for $s < t < \inf U$ with $|t-s| \leq \delta$, 
    $\cov(X_s, X_t \mid X_U)$ is non zero.
\end{lemma}

\begin{proof}[Proof of \Cref{lemma:slnd}]
    Let $w$ be the modulus of continuity of the covariance kernel of $X$, and let $u \coloneq \min U$. We then have for $s< t < \min U$:
    \begin{align*}
        \cov(X_s, X_t \mid X_U) &= \frac{1}{2} \left[\var(X_s \mid X_U) + \var(X_t \mid X_U) - \var(X_t - X_s \mid X_U)  \right] \\
        &\geq \frac{1}{2} \left[ c_0 \phi(u-s) + c_0 \phi(u-t) - w(s-t) \right]
    \end{align*}
    As $t-s \to 0^+$, $w(t-s) \to 0$ while $u-t>\delta$ for some $\delta>0$ and $\phi(u-s) > \epsilon$, $\phi(u-t) > \epsilon$ for some $\epsilon > 0$.
    It follows that we can find $\delta'$ such that for all $s, t$ satisfying $|t-s| \leq \delta'$, then $\cov(X_s, X_t \mid X_U) > 0$.
\end{proof}

\begin{proof}[Proof of \Cref{thm:endpoints-cond-ind}]
    It is clear that if $X$ is Markov then \Cref{eq:endpoints-cond-ind} holds.
    For the converse, we adopt the perspective of continuously indexed graphical models \citep{waghmare2025continuously}
    and we show that if $X$ is not Markov, then
    \Cref{eq:endpoints-cond-ind} cannot hold. 
    An illustration of the proof is given in \Cref{fig:jumping-path}.
    Assume that $X$ does not satisfy the Markov property. 
    Then, there exist $s < t$ and a closed interval $U$ included in $(s,t)$
    such that $X_s \notind X_t \mid X_U$.
    Let $\Omega$ be a graphical model for $X$.
    By definition of $\Omega$ (see \Cref{remark:notind-path}), 
    $s$ and $t$ cannot be separated by $U$ in $\Omega$.
    Therefore, we can find a path $\bp(s,t)$ in $\Omega$ from $s$ to $t$
    such that this path does not cross $U$.
    
    Furthermore, since $X$ has continuous sample paths, the covariance
    kernel $K$ of $X$ is continuous around the diagonal \citep{cambanis1973some}.
    Without loss of generality, we can assume that the set on which $K$ is continuous 
    around the diagonal is closed. Therefore, it is compact. 
    As the process is not degenerate, $K$ is non-null on the diagonal.
    Let $\epsilon = \min_{t\in[a,s]} K(t,t) > 0$.
    By Heine-Cantor theorem, 
    we have $\delta > 0$ such that $K$ satisfies
    \begin{equation*}
        |u-v| \leq \delta \implies |K(u,v)| \geq \epsilon \quad \text{for all } u,v \in [a,s].
    \end{equation*}
    It follows that $X_u \notind X_v$ for all $u, v \in [a,s]$ 
    such that $|u-v| \leq \delta$.
    Let $\delta' > 0$ given by \Cref{lemma:slnd}. 
    Let $(u_1, \dots, u_n)$ be a strictly increasing sequence with $n \geq 1$ such that
    $u_1 = a$, $u_n = s$ and $|u_{k+1} - u_k| \leq \delta \wedge \delta'$ for all $k$.
    Let $k \in \{1, \dots, n-1\}$.
    As $X_{u_k} \notind X_{u_{k+1}}$, there is a path between $u_k$
    and $u_{k+1}$.
    Furthermore, \Cref{lemma:slnd}, $X_{u_k} \notind X_{u_{k+1}} \mid X_U$ :
    this gives us a path between $u_k$ and $u_{k+1}$
    not crossing $U$ (see \Cref{remark:notind-path}).
    Concatenating these paths for all $1 \leq k \leq n-1$,
    we obtain a path $\bp(a,s)$ not crossing $U$.
    We apply the same approach to build a path $\bp(t,b)$ not crossing $U$.

    We can thus concatenate $\bp(a,s)$, $\bp(s,t)$ and $\bp(t,b)$ to obtain
    a path from $a$ to $b$ not crossing $U$. We conclude by \Cref{thm:non-singular-correlation}.
\end{proof}

\begin{figure}[h!]
    \centering
    \includegraphics[width=0.30\textwidth]{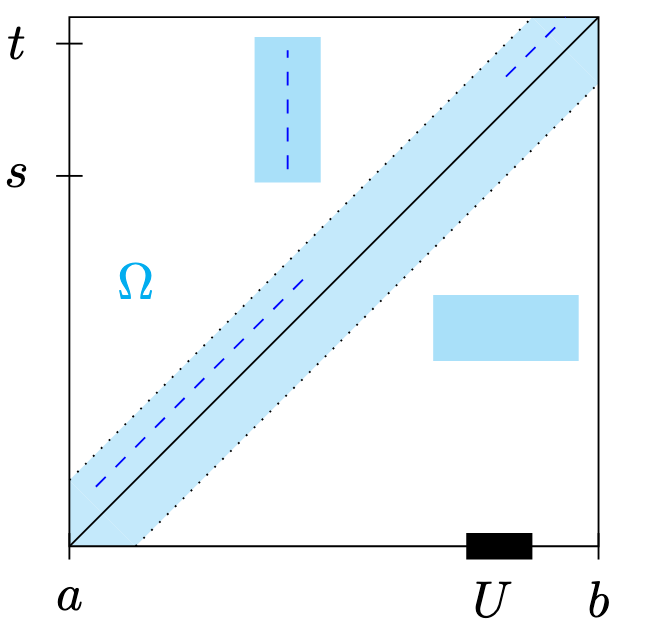}
    \caption{Graphical model with path (in dashed blue) going from $a$ to $b$ 
    without intersecting $U$.}
    \label{fig:jumping-path}
\end{figure}

\begin{proof}[Proof of \Cref{prop:cov-endpoints}]
    Assume $X$ is Markov on $[a,b]$. Then we know that there exist functions $f$ and $g$ such that 
    $K(s,t)=f(s)g(t)$ for all $s,t\in [a,b]$ such that $s\leq t$ (see \Cref{eq:kernel-decomp}).
    Then, we have $K(a,b) = f(a)g(b) > 0$ since $X$ is not degenerate.
    Moreover, 
    \begin{align*}
        g(t) = \frac{1}{f(a)}K(a,t) \quad \text{for all } t \geq a,
        f(s) = \frac{1}{g(b)}K(s,b) \quad \text{for all } s \leq b.
    \end{align*}
    As $\frac{1}{f(a)}\cdot\frac{1}{g(b)} = \frac{1}{K(a,b)}$, it follows that
    $K(s,t) = \frac{1}{K(a,b)}K(s,b)K(a,t)$ for all $a \leq s \leq t \leq b$.

    The converse is clear by \Cref{prop:kernel-functional-equation} and \Cref{eq:kernel-decomp}.
\end{proof}

\begin{proof}[Proof of \Cref{thm:kernel-endpoints}]
    Assume that $X$ is Markov on $[a,b]$. 
    Recall that, for $a \leq s \leq t \leq b$, we have by \Cref{prop:cov-endpoints}:
    \begin{equation}
        K(s,t) = \frac{K(s,u)K(u,t)}{K(u,u)} \quad \text{for all } u \in [s,t].
    \end{equation}
    Taking $s=a$ and $t=b$ gives the desired expression.

    Conversely, assume that $K(a,b) = {K(a,u)K(u,b)}/{K(u,u)}$ holds for all $u \in [a,b]$.
    As $X$ is Gaussian, this is equivalent to having:
    \begin{equation*}
    \begin{aligned}
    \text{i.e.} && &\cov\big( X_a, X_b \mid X_u\big)= 0, \\
    \text{i.e.} && & X_a \ind X_b \mid X_u.
    \end{aligned}
    \end{equation*}
    We know by \Cref{thm:endpoints-cond-ind} that 
    $X_a \ind X_b \mid X_u, \: \forall u\in(a,b)$ implies that $X$ is Markov on $[a,b]$.

\end{proof}

\paragraph{Finite-dimensional endpoints characterization}

In finite-dimensional settings, non-singularity of the correlation structure and non-determinism are not needed, and we can simply suppose that the vector has non-degenerate Gaussian distribution.
For stochastic processes, continuity of the sample paths implies a local dependence of nearby values. 
We translate this property at the discrete level in terms of conditional independence.
Specifically, assume that $X$ is a non-degenerate Gaussian 
random vector of size $p \geq 2$. 
We say that the coordinates of $X$ are \emph{locally
dependent} if 
$X_j \notind X_{j+1} \mid X_{\{j,j+1\}^c}$
for all $j \in [p-1]$.
We can then
characterize the Markov property of $X$ as follows.
\begin{theorem}\label{thm:discrete-endpoints-cond-ind}
    Let $X = (X_1, \dots, X_p) \in \RR^p$ be a non-degenerate Gaussian vector,
    and assume that the coordinates of $X$ are locally dependent.
    Then, $X$ satisfies the Markov property if and only if 
    \begin{equation}\label{eq:discrete-endpoints-cond-ind}
        X_1 \ind X_p \mid X_j \quad \text{for all } j\in \{2, \dots, p-1\}.
    \end{equation}
\end{theorem}
\begin{proof}
    It is clear that if $X$ is Markov then \Cref{eq:discrete-endpoints-cond-ind} holds.
    For the converse, assume that $X$ is not Markov and denote by $\OmX$ the conditional independence graph of $X$ (see \Cref{apx:graphical-models}).
    Then we have a triplet $i < j < k$ such that $X_i \notind X_k \mid X_j$.
    This implies that there is a path $\bp(i,k)$ from $X_i$ to $X_k$ not crossing $X_j$ in $\OmX$.
    Furthermore, as $X$ is locally dependent, we know that there is an edge between $X_l$ and $X_{l+1}$ in $\OmX$ for all $l\in[p-1]$.
    This yields a path $\bp(1,i)$ from $X_1$ to $X_i$ and a path $\bp(k,p)$ from $X_k$ to $X_p$. Concatenating the three paths $\bp(1,i), \bp(i,k)$ and $\bp(k,p)$
    we get a path from $X_1$ to $X_p$. 
    Moreover, as $i < j < k$, the paths $\bp(1,i)$ and $\bp(k,p)$ do not cross $X_j$, and by definition of $\bp(i,k)$ it does not cross $X_j$ either.
    Hence there is a path from $X_1$ to $X_p$ not crossing $X_j$ in the graph $\OmX$. It follows that $X_1 \notind X_p \mid X_j$, and concludes the proof.
\end{proof}

\paragraph{Test convergence rates}

\begin{proof}[Proof of \Cref{thm:test-rate}]
    We follow proof ideas introduced in \citep{shah2018hardness}.
    We define
    \begin{align*}
    \widetilde{T}_n = \max_{j \in \{2, \dots, p-1\}} \sqrt{n} \cdot |z(\widetilde{\rho}_j) - z_j|
    \quad \text{where } \widetilde{\rho}_j = 
    \frac{
    \frac{1}{n}\sum_{k=1}^n e_{j1}^{(k)} e_{pj}^{(k)}}{\sigma_{j1}\sigma_{pj}},
    \end{align*}
    with $\sigma_{j1}, \sigma_{pj}$ the limiting standard deviations of the residuals $e_{j1}, e_{pj}$.
    By the triangle inequality, 
    \begin{align}\label{proof:test-cv-triangle-ineq}
        \sup_{t \geq 0} \left| \PP \left\{\widehat{T}_n \leq t \right\} - \PP \left\{ \max_{1<j<p} |\widehat{Y}_j| \leq t \right\} \right| 
        &\leq 
        \sup_{t \geq 0} \left| \PP \left\{\widehat{T}_n \leq t \right\} - \PP \left\{ \widetilde{T}_n \leq t \right\} \right| \\
        & \qquad+ 
        \sup_{t \geq 0} \left| \PP \left\{\widetilde{T}_n \leq t \right\} - \PP \left\{ \max_{1<j<p} |\widehat{Y}_j| \leq t \right\} \right|.
    \end{align}
    Using $|\PP(A)-\PP(B)| \leq \PP(A \setminus B) + \PP(B \setminus A)$, we can write 
    for all $t \geq 0$:
    \begin{align*}
        \left| \PP \left\{\widehat{T}_n \leq t \right\} - \PP \left\{ \widetilde{T}_n \leq t \right\} \right|
        &\leq 
        \PP\left( \sigma_{j1} \sigma_{pj} x < \hat{S}_j \leq \hat{\sigma}_{j1} \hat{\sigma}_{pj} x, \; \forall j \right)
        + \PP\left( \hat{\sigma}_{j1} \hat{\sigma}_{pj}x < \hat{S}_j \leq {\sigma}_{j1} {\sigma}_{pj}x, \; \forall j \right) \\
        &= \PP\left( \left|\hat{S}_j/x - \sigma_{j1} \sigma_{pj} \right| < \left| \hat{\sigma}_{j1} \hat{\sigma}_{pj} - \sigma_{j1} \sigma_{pj}  \right|, \; \forall j \right) \\
        & \leq \PP\left( \left|\hat{S}_j/x - \sigma_{j1} \sigma_{pj} \right| < \left| \hat{\sigma}_{j1} \hat{\sigma}_{pj} - \sigma_{j1} \sigma_{pj}  \right| \right),
    \end{align*}
    for a fixed $j$
    and where $\hat{\sigma}_{j1}, \hat{\sigma}_{pj}$ are the empirical standard deviations
    of the corresponding residuals, $\hat{S}_j$ is the empirical product of residuals, i.e.
    $\hat{S}_j = \frac{1}{n} \sum_{k=1}^{n} e_{j1}^{(k)}e_{pj}^{(k)}$, and $x = \tanh(t)$.
    Moreover, we write for brevity $\Delta_n = |\hat{\sigma}_{j1} \hat{\sigma}_{pj} - \sigma_{j1} \sigma_{pj}|$ and $\widetilde{S} = \hat{S}_j/x - \sigma_{j1} \sigma_{pj}$.
    Then, by law of iterated expectations,
    \begin{align*}
        \PP\left[ \left|\hat{S}_j/x - \sigma_{j1} \sigma_{pj} \right| < \left| \hat{\sigma}_{j1} \hat{\sigma}_{pj} - \sigma_{j1} \sigma_{pj}  \right| \right]
        &= \EE \left[ \PP\left( \left|\widetilde{S}\right| < \Delta_n \Big| \Delta_n \right)  \right]
        \\
        &= \EE \left[ F_{\widetilde{S}| \Delta_n}(\Delta_n \big| \Delta_n) - F_{\widetilde{S}| \Delta_n}(-\Delta_n \big| \Delta_n) \right] \\
        &= \EE \left[ 2 f_{\widetilde{S}|\Delta_n}(0 \big| \Delta_n) \Delta_n \right] +  o(\EE\Delta_n) \\
        &\leq M \EE[\Delta_n] + o(\Delta_n),
    \end{align*}
    for some constant $M>0$ as $\widetilde{S}$ and $\Delta_n$ converge as $n$ grows.
    Moreover, by the central limit theorem and the delta method,
    $\EE \left[ \Delta_n \right]
    = C/ \sqrt{n} + o\left( {1}/{\sqrt{n}} \right)$
    where $C$ is a positive constant.
    It follows that:
    \begin{align}\label{eq:proof-test-cv-term1}
        \left| \PP \left\{\widehat{T}_n \leq t \right\} - \PP \left\{ \widetilde{T}_n \leq t \right\} \right| \leq CM \frac{1}{\sqrt{n}} + o\left( \frac{1}{\sqrt{n}} \right).
    \end{align}

    We now bound the second term of \Cref{proof:test-cv-triangle-ineq} using
    the results of \citep{chernozhukov2013gaussian}.
    By a change of variable, we can restrict our attention
    to $\widetilde{\rho}$.
    Observe that $\widetilde{\rho}_j$ is the sum of products of two mixture of Gaussian
    distributions. In particular, each of this product is sub-exponential, 
    and by continuity of the process over the considered interval we can find $K > 0$
    such that the subexponential norm $\norm{\cdot}_{\psi_1} \coloneqq \inf \{K'>0: \; \EE \exp(|\cdot|/K') \leq 2 \}$ of all the products $e_{j1} e_{pj}$ is not larger than $K$.
    Then, we have by a union bound that:
    \begin{align*}
        \PP \left( \max_{j} |e_{j1} e_{pj}| > t \right) \leq 2p e^{-t/K}
    \end{align*}
    and, taking $t = K \log p$ we obtain that
    \begin{align*}
        \norm{\max_{1<j<p} |e_{j1} e_{pj}|}_{\psi_1} \leq K (\log p +1).
    \end{align*}
    Moreover, there are constants $c_1, C_1>0$ such that $c_1 \leq \EE[(e_{j1} e_{pj})^2] \leq C_1$ by continuity of the process on the
    interval and since its covariance is never null.
    Using \textit{Theorem 2.2} and \textit{Corollary 2.2} in \citep{chernozhukov2013gaussian},
    we obtain
    \begin{align*}
        \sup_{t \geq 0} \left| \PP \left\{\widetilde{T}_n \leq t \right\} - \PP \left\{ \max_{1<j<p} |\widehat{Y}_j| \leq t \right\} \right|
        \leq C \left[ \frac{1}{n^{1/8}} \left\{ \log(pn/\gamma) \right\}^{7/8} + \frac{1}{n^{1/2}} \left\{ \log(pn/\gamma) \right\}^{3/2} u(\gamma) + \gamma \right],
    \end{align*}
    for all $\gamma \in (0,1)$ and where $u(\gamma) \leq C' \log(p) \log(n/\gamma + 1)$, for some
    $C, C'>0$.
    Taking $\gamma = n^{-1/2}$ yields
    \begin{align}\label{eq:proof-test-cv-term2}
        \sup_{t \geq 0} &\left| \PP \left\{\widetilde{T}_n \leq t \right\} - \PP \left\{ \max_{1<j<p} |\widehat{Y}_j| \leq t \right\} \right|
        \lesssim  \frac{1}{n^{1/8}} \left\{ \log(pn^{3/2}) \right\}^{7/8} + \frac{1}{n^{1/2}} \left\{ \log(p n^{3/2}) \right\}^{3/2} \log p \log n.
    \end{align}
    Assume that $p = \exp(n^\alpha)$, for some $\alpha > 0$.
    Then, the derived upper-bound converges
    if and only if $\alpha \in (0, 1/7)$,
    and, in this case, the first term of the above equation
    is dominant.
    We conclude by putting together 
    \Cref{eq:proof-test-cv-term1} and \Cref{eq:proof-test-cv-term2}.
\end{proof}

\section{Additional numerical illustrations}\label{apx:numerics}

\paragraph{Kernel-embedded Brownian motion.}
We consider Gaussian processes obtained by embedding the covariance of a Brownian motion into a compactly supported kernel with bandwidth $h>0$. This construction produces processes with $h$-long memory in the sense that conditioning on an interval longer than $h$ renders past and future approximately independent. Formally, let $K_0$ be a covariance kernel and let $k:[0,1]^2\to\mathbb{R}_+$ be a bivariate kernel. Define the rescaled kernel $k_h(\cdot)=(1/h)k(\cdot/h,\cdot/h)$. The resulting kernel-embedded covariance is
\begin{align*}
K_h(s,t)=\iint_{[0,1]^2} k_h(s,u)\,K_0(u,v)\,k_h(v,t)\,du\,dv, \qquad s,t\in[0,1].
\end{align*}
In our simulations we take $K_0$ to be the Brownian motion covariance and choose $k_h$ to be a Wendland kernel. Specifically, we use the compactly supported radial function $\kappa(r)=\max(1-r, 0)^{4}(4r+1)$ for $r\ge 0$ and define
\begin{align*}
    k_h(s,u)=\frac{1}{h}\,\kappa\!\left(\frac{|s-u|}{h}\right), \quad \text{for } s,u\in[0,1].
\end{align*}
Wendland kernels are compactly supported and positive definite, which ensures that the embedded kernel $K_h$ remains a valid covariance kernel whenever $K_0$ is. Their compact support also guarantees locality of dependence: $k_h(s,\cdot)$ is supported in a neighborhood of radius $h$, so the parameter $h$ directly controls the range of dependence and the deviation from the Markov property. Sample paths of the resulting process are shown for varying bandwidths in \Cref{fig:samples-kebm}. Smoothing the Brownian covariance along both arguments increases regularity near the diagonal, and the resulting sample paths become smoother as the bandwidth grows. This framework allows us to systematically compare Markov and non-Markov processes while controlling the degree of dependence.

\begin{figure}[h!]
    \centering
    \begin{minipage}{0.3\textwidth}
        \centering
        \includegraphics[width=\textwidth]{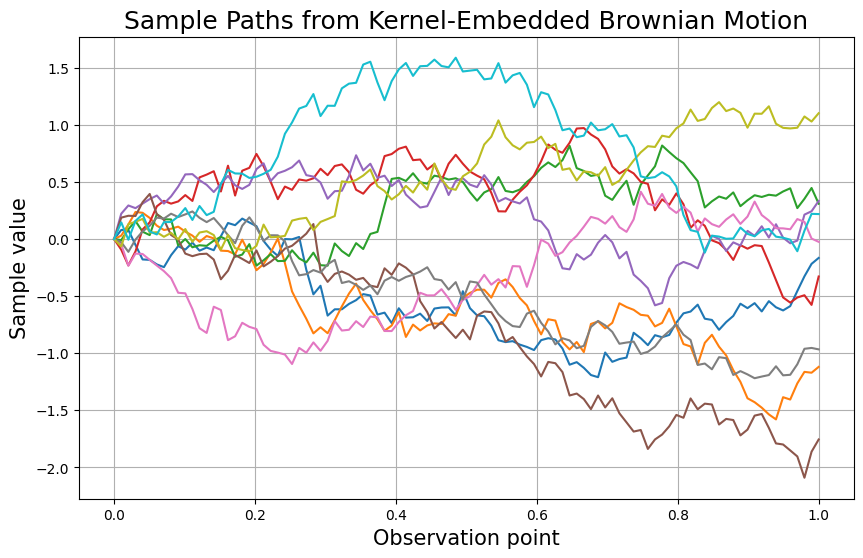}
    \end{minipage}
    \hspace{0.02\textwidth}
    \begin{minipage}{0.3\textwidth}
        \centering
        \includegraphics[width=\textwidth]{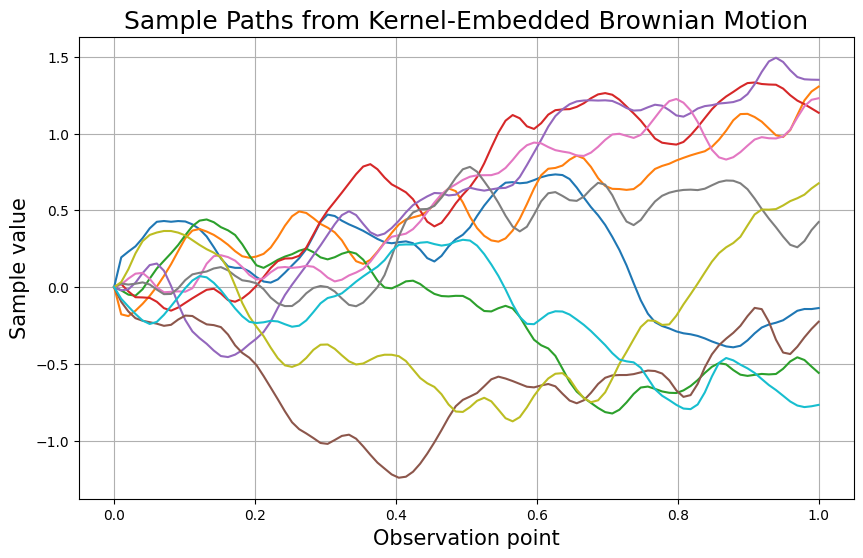}
    \end{minipage}
    \hspace{0.02\textwidth}
    \begin{minipage}{0.3\textwidth}
        \centering
        \includegraphics[width=\textwidth]{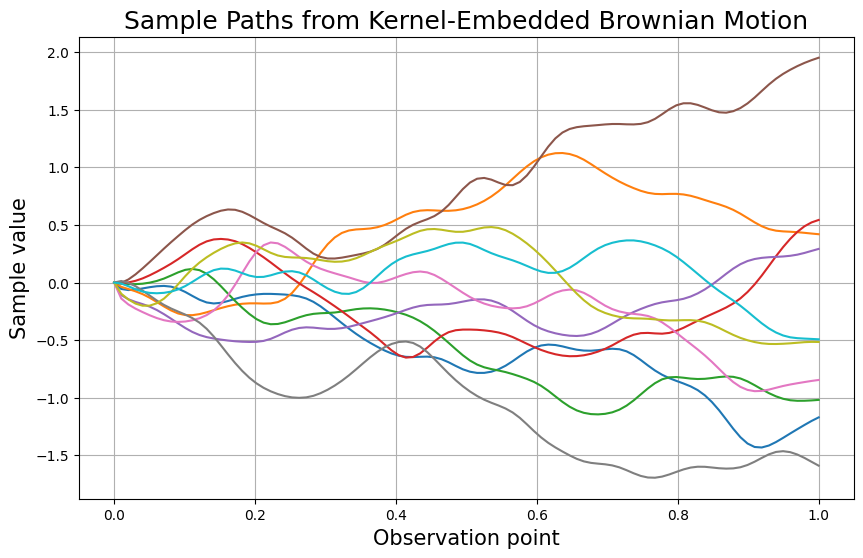}
    \end{minipage}

    \caption{Samples of the kernel-embedded Brownian motion with $p=100$ and bandwidth $h=0.01, 0.05, 0.1$.}
    \label{fig:samples-kebm}
\end{figure}

~
\newpage

\subsection{Convergence}

\begin{figure}[h!]
    \centering
    \begin{minipage}{0.32\textwidth}
        \centering
        \includegraphics[width=\textwidth]{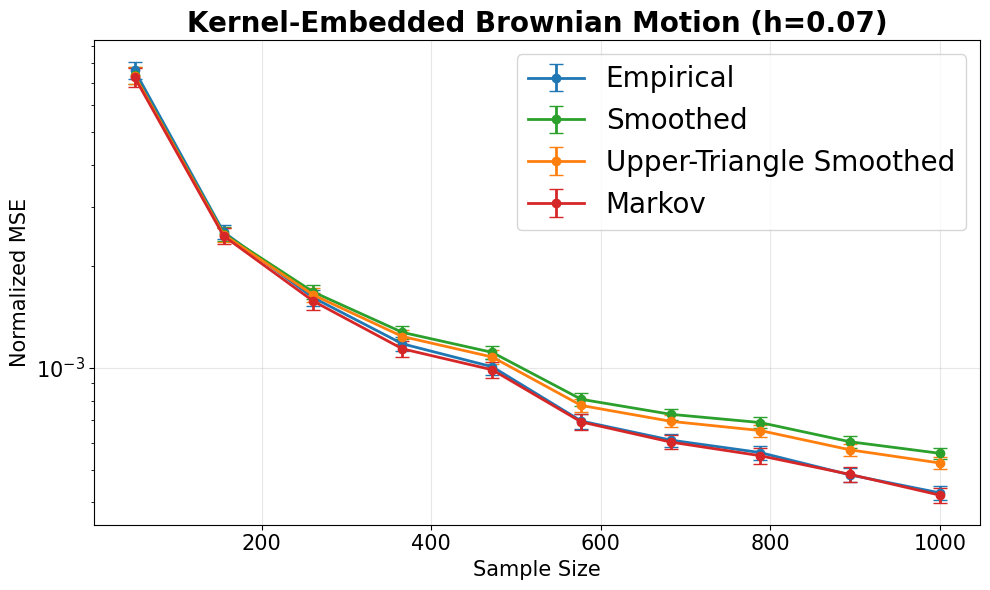}
    \end{minipage}
    \begin{minipage}{0.32\textwidth}
        \centering
        \includegraphics[width=\textwidth]{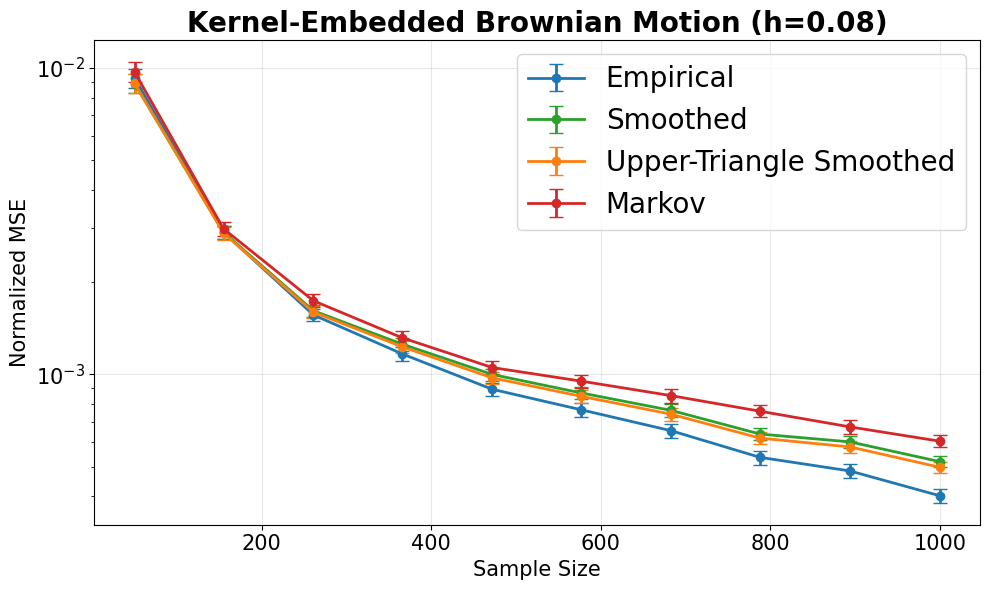}
    \end{minipage}
    \begin{minipage}{0.32\textwidth}
        \centering
        \includegraphics[width=\textwidth]{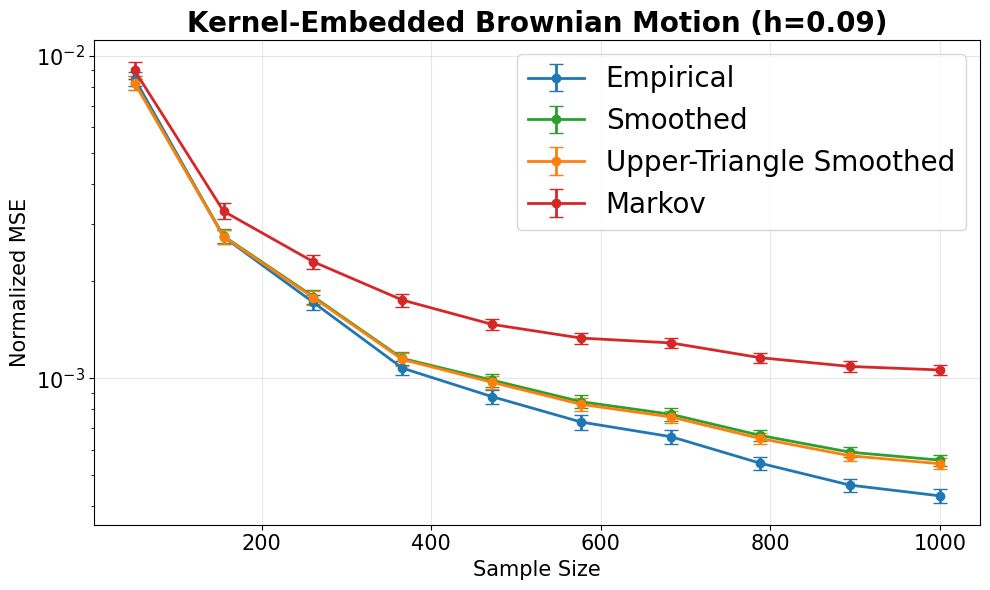}
    \end{minipage}

    \caption{Comparison of convergence in $L_2$-norm as $n$ grows for the kernel-embedded Brownian motion with $h=0.07, 0.08, 0.09$.}
    \label{fig:L2-transition}
\end{figure}

\begin{figure}[h!]
    \centering
    \includegraphics[width=\textwidth]{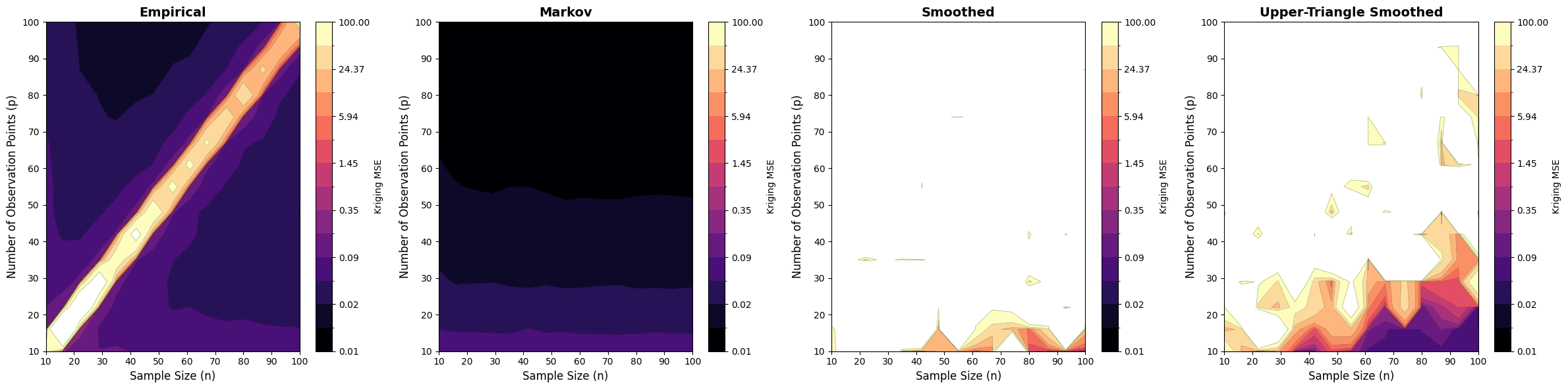}
    \caption{Comparison of kriging mean square errors for the Ornstein-Uhlenbeck process as $n$ and $p$ grow.}
    \label{fig:kriging-contour-ou}
\end{figure}

\begin{figure}[h!]
    \centering
    \includegraphics[width=\textwidth]{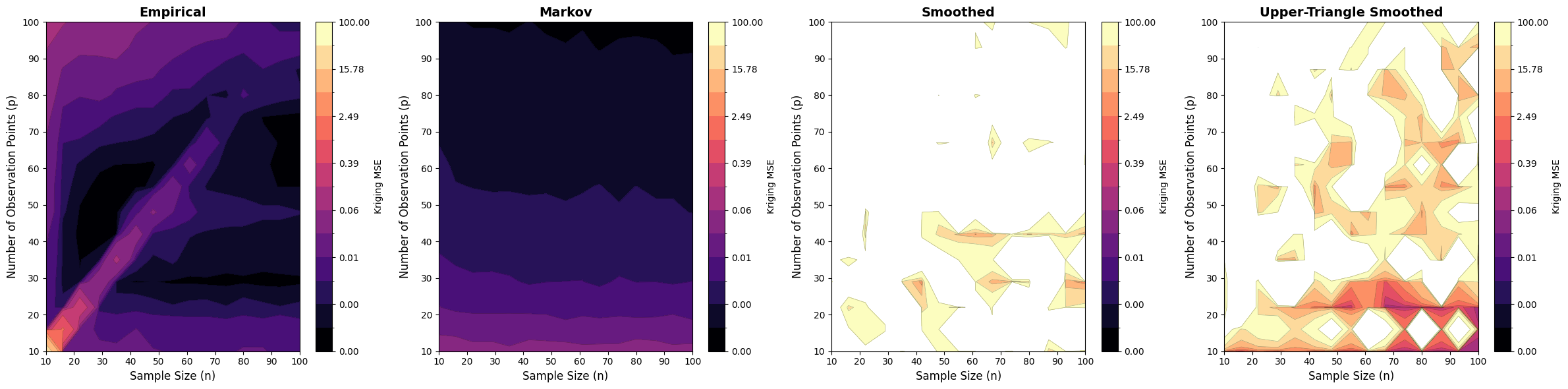}
    \caption{Comparison of kriging mean square errors for the kernel-embedded Brownian motion ($h=0.2$) as $n$ and $p$ grow.}
    \label{fig:kriging-contour-kebm}
\end{figure}

\newpage

\subsection{Testing}

\begin{figure}[h!]
\centering
\includegraphics[height=0.22\textheight]{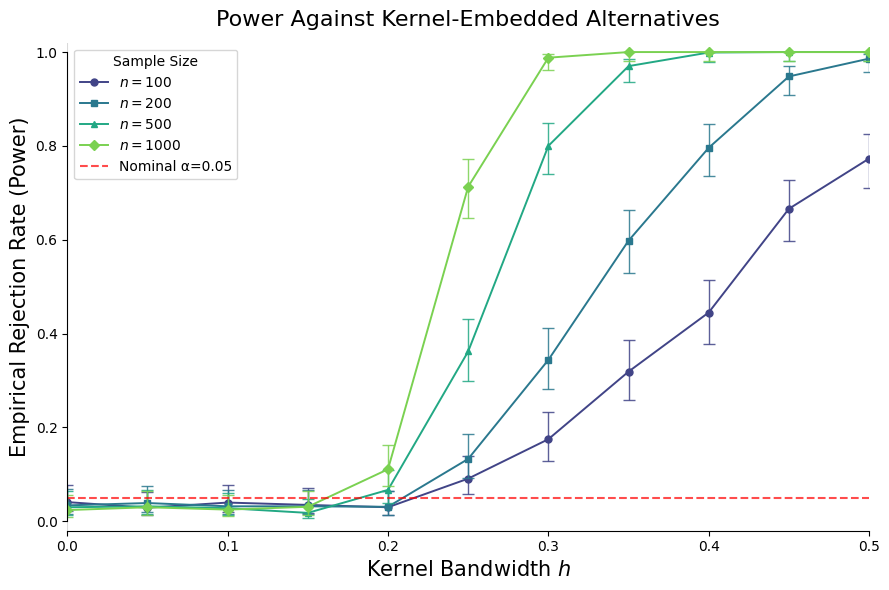}
\includegraphics[height=0.22\textheight]{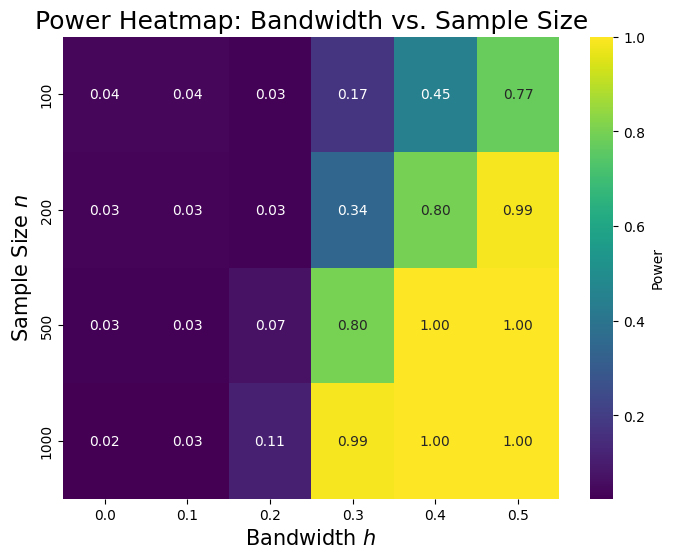}
\caption{Power of the proposed test against local alternatives generated by kernel-embedded Brownian motions with bandwidth $h \in [0,0.5]$, for $p=10$.}
\end{figure}

\begin{figure}[h!]
\centering
\includegraphics[height=0.22\textheight]{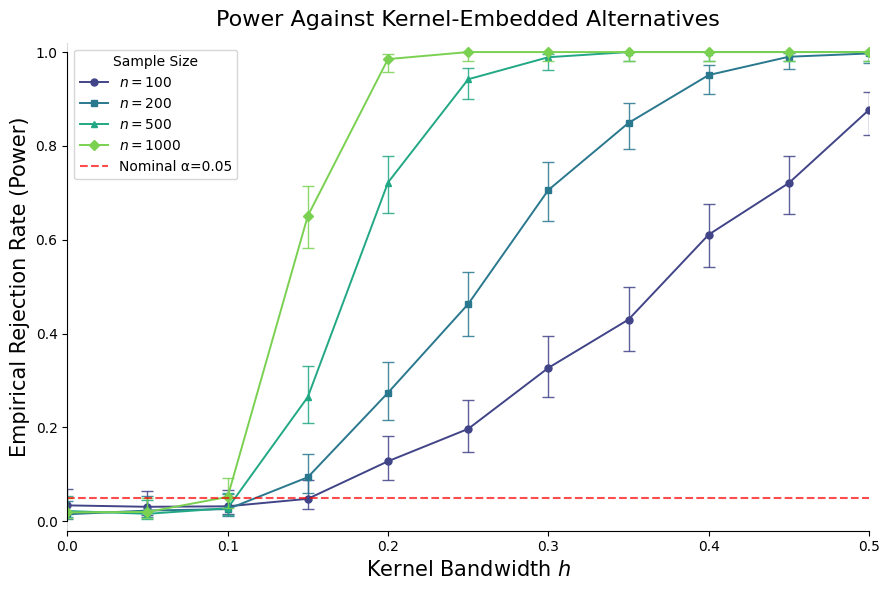}
\includegraphics[height=0.22\textheight]{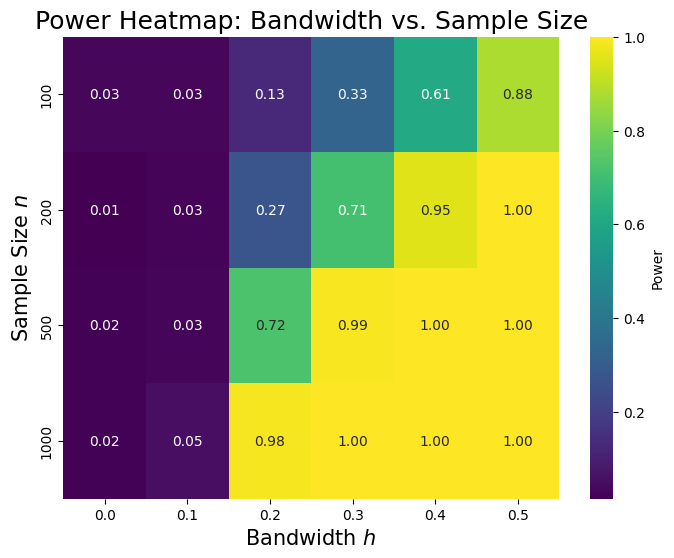}
\caption{Power of the proposed test against local alternatives generated by kernel-embedded Brownian motions with bandwidth $h \in [0,0.5]$, for $p=20$.}
\end{figure}

\begin{figure}[h!]
\centering
\includegraphics[height=0.22\textheight]{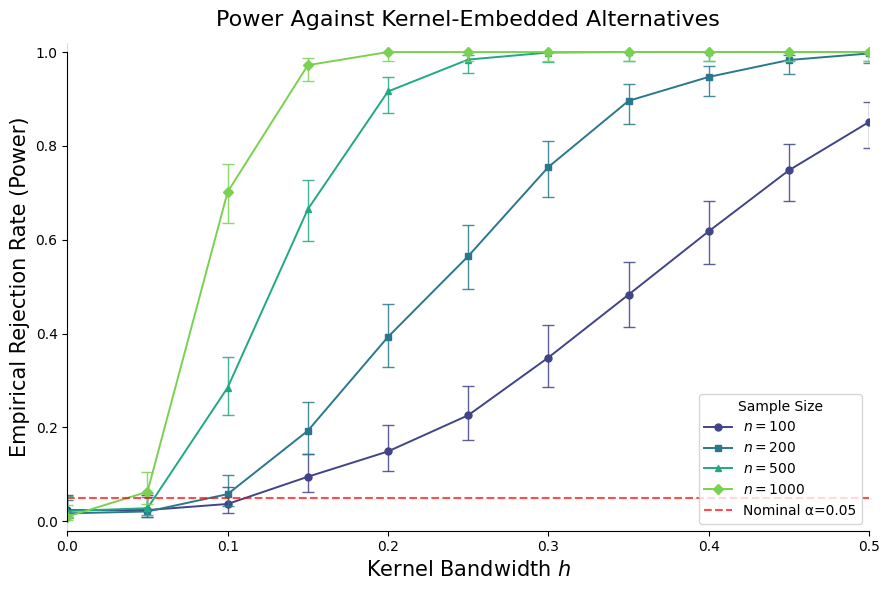}
\includegraphics[height=0.22\textheight]{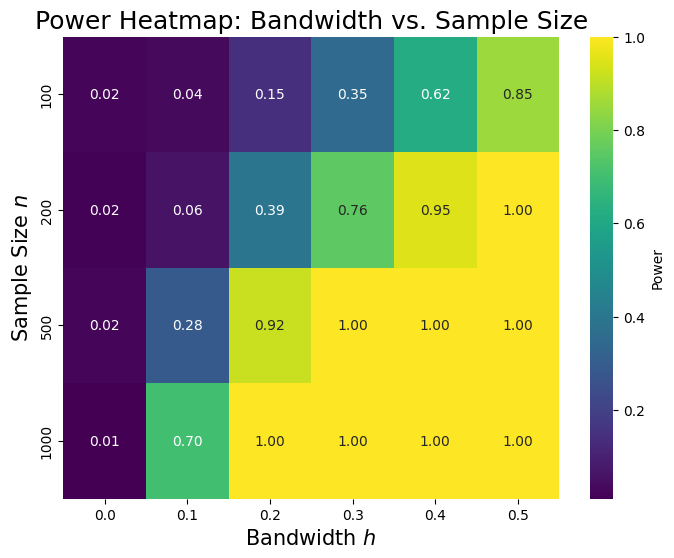}
\caption{Power of the proposed test against local alternatives generated by kernel-embedded Brownian motions with bandwidth $h \in [0,0.5]$, for $p=50$.}
\end{figure}

\begin{figure}[h!]
    \centering
    \begin{minipage}{0.32\textwidth}
        \centering
        \includegraphics[width=\textwidth]{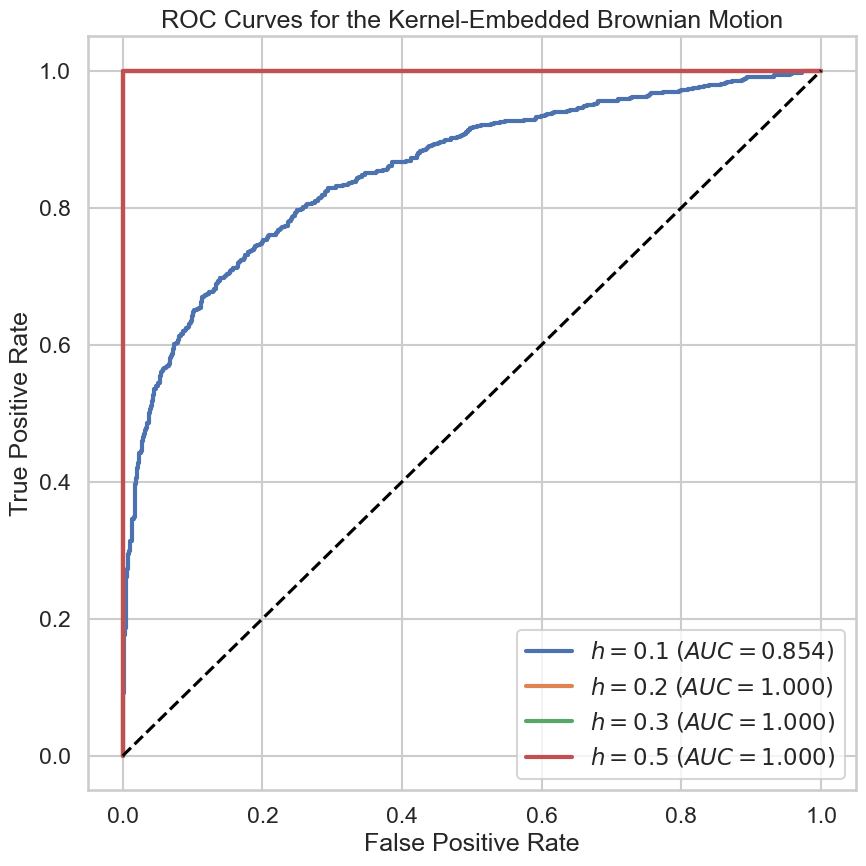}
    \end{minipage}
    \begin{minipage}{0.32\textwidth}
        \centering
    \includegraphics[width=\textwidth]{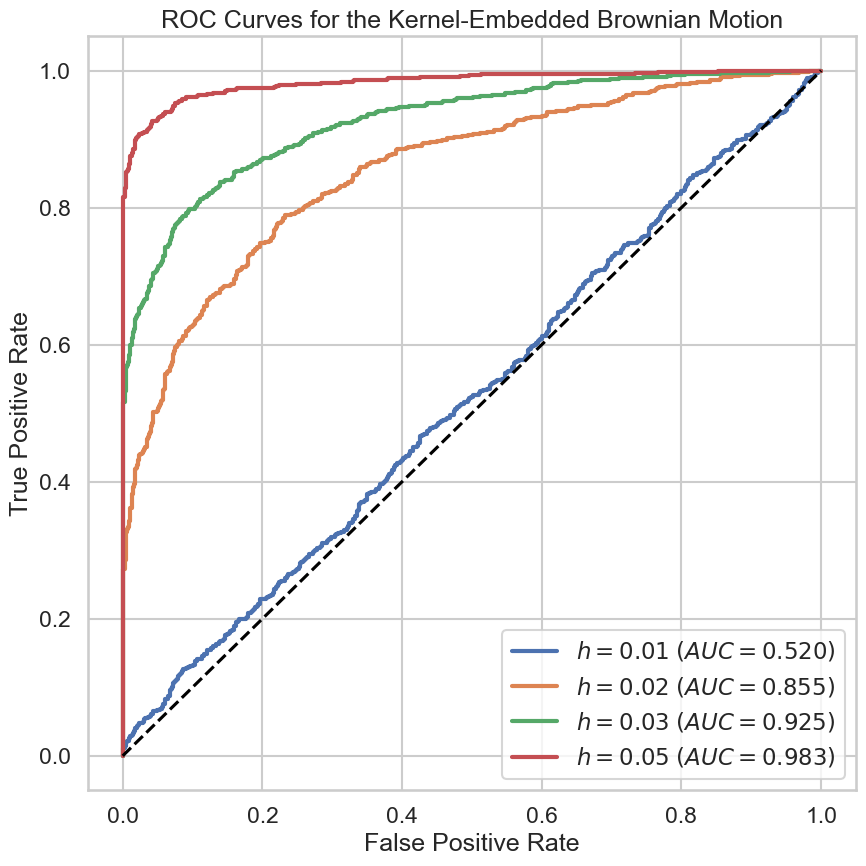}
    \end{minipage}
    \begin{minipage}{0.32\textwidth}
        \centering
        \includegraphics[width=\textwidth]{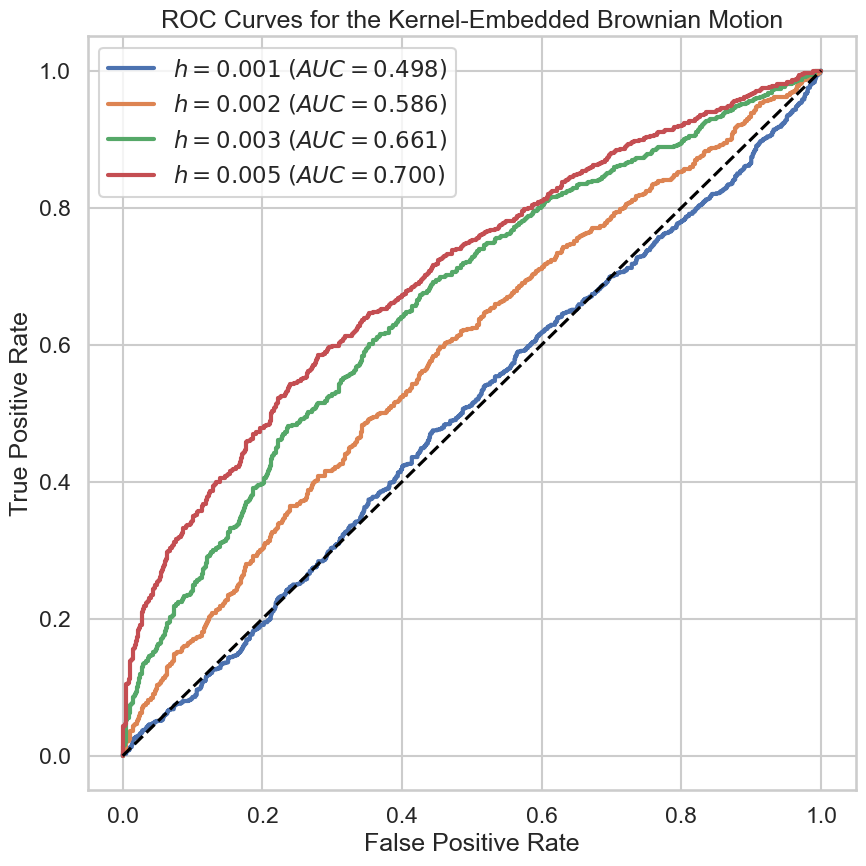}
    \end{minipage}
    \caption{ROC curves for the kernel-embedded Brownian motion with $p =10, 100, 1000$.}
    \label{fig:testing-roc-kebm}
\end{figure}

%% file: references.bib
@book{doob1962stochastic,
  title={Stochastic Processes},
  author={Doob, Joseph L.},
  series={Wiley publications in statistics},
  year={1962},
  publisher={Wiley}
}

@article{mehr1965,
    author = {Mehr, Cyrus B. and McFadden, James A.},
    title = "{Certain Properties of Gaussian Processes and Their First-Passage Times}",
    journal = {Journal of the Royal Statistical Society: Series B (Methodological)},
    volume = {27},
    number = {3},
    pages = {505-522},
    year = {1965},
    issn = {0035-9246},
}

@article{mohammadi2024rough,
  title={Functional data analysis with rough sample paths?},
  author={Mohammadi, Neda and Panaretos, Victor M.},
  journal={Journal of Nonparametric Statistics},
  volume={36},
  number={1},
  pages={4--22},
  year={2024},
  publisher={Taylor \& Francis}
}

@article{delaigle2016approximating,
  title={Approximating fragmented functional data by segments of Markov chains},
  author={Delaigle, Aurore and Hall, Peter},
  journal={Biometrika},
  volume={103},
  number={4},
  pages={779--799},
  year={2016},
  publisher={Oxford University Press}
}

@article{waghmare2025continuously,
  title={Continuously indexed graphical models},
  author={Waghmare, Kartik G. and Panaretos, Victor M.},
  journal={Journal of the Royal Statistical Society Series B: Statistical Methodology},
  volume={87},
  number={1},
  pages={211--231},
  year={2025},
  publisher={Oxford University Press UK}
}

@article{waghmare2022completion,
  title={The completion of covariance kernels},
  author={Waghmare, Kartik G. and Panaretos, Victor M.},
  journal={The Annals of Statistics},
  volume={50},
  number={6},
  pages={3281--3306},
  year={2022},
  publisher={Institute of Mathematical Statistics}
}

@article{cambanis1973some,
  title={On some continuity and differentiability properties of paths of Gaussian processes},
  author={Cambanis, Stamatis},
  journal={Journal of Multivariate Analysis},
  volume={3},
  number={4},
  pages={420--434},
  year={1973},
  publisher={Elsevier}
}

@article{yao2005functional,
  title={Functional data analysis for sparse longitudinal data},
  author={Yao, Fang and M{\"u}ller, Hans-Georg and Wang, Jane-Ling},
  journal={Journal of the American statistical association},
  volume={100},
  number={470},
  pages={577--590},
  year={2005},
  publisher={Taylor \& Francis}
}

@book{hsing2015theoretical,
  title={Theoretical foundations of functional data analysis, with an introduction to linear operators},
  author={Hsing, Tailen and Eubank, Randall},
  year={2015},
  publisher={John Wiley \& Sons}
}

@book{ramsay2005,
  author = {Ramsay, James O. and Silverman, Bernard W.},
  isbn = {9780387400808},
  publisher = {Springer},
  title = {{Functional Data Analysis}},
  year = 2005
}

@article{alghattas2024covariance,
  title={Covariance Operator Estimation via Adaptive Thresholding},
  author={Al-Ghattas, Omar and Sanz-Alonso, Daniel},
  journal={arXiv preprint arXiv:2405.18562},
  year={2024}
}

@article{shah2018hardness,
author = {Shah, Rajen and Peters, Jonas},
year = {2018},
month = {04},
pages = {},
title = {The Hardness of Conditional Independence Testing and the Generalised Covariance Measure},
volume = {48},
journal = {The Annals of Statistics},
doi = {10.1214/19-AOS1857}
}

@article{chen2012testing,
  title={Testing for the Markov property in time series},
  author={Chen, Bin and Hong, Yongmiao},
  journal={Econometric Theory},
  volume={28},
  number={1},
  pages={130--178},
  year={2012},
  publisher={Cambridge University Press}
}

@article{zhou2023testing,
  title={Testing for the Markov property in time series via deep conditional generative learning},
  author={Zhou, Yunzhe and Shi, Chengchun and Li, Lexin and Yao, Qiwei},
  journal={Journal of the Royal Statistical Society Series B: Statistical Methodology},
  volume={85},
  number={4},
  pages={1204--1222},
  year={2023},
  publisher={Oxford University Press US}
}

@article{baker1973joint,
  title={Joint measures and cross-covariance operators},
  author={Baker, Charles R.},
  journal={Transactions of the American Mathematical Society},
  volume={186},
  pages={273--289},
  year={1973}
}

@article{qiao2019functional,
  title={Functional graphical models},
  author={Qiao, Xinghao and Guo, Shaojun and James, Gareth M},
  journal={Journal of the American Statistical Association},
  volume={114},
  number={525},
  pages={211--222},
  year={2019},
  publisher={Taylor \& Francis}
}

@article{qiao2020doubly,
  title={Doubly functional graphical models in high dimensions},
  author={Qiao, Xinghao and Qian, Cheng and James, Gareth M and Guo, Shaojun},
  journal={Biometrika},
  volume={107},
  number={2},
  pages={415--431},
  year={2020},
  publisher={Oxford University Press}
}

@article{yuan2007glasso,
  title={Model selection and estimation in the Gaussian graphical model},
  author={Yuan, Ming and Lin, Yi},
  journal={Biometrika},
  volume={94},
  number={1},
  pages={19--35},
  year={2007},
  publisher={Oxford University Press}
}

@article{chernozhukov2013gaussian,
  title={Gaussian approximations and multiplier bootstrap for maxima of sums of high-dimensional random vectors},
  author={Chernozhukov, Victor and Chetverikov, Denis and Kato, Kengo},
  journal={The Annals of Statistics},
  pages={2786--2819},
  year={2013},
  publisher={JSTOR}
}

@article{chao1972negative,
  title={Negative moments of positive random variables},
  author={Chao, Min-Te and Strawderman, William E.},
  journal={Journal of the American Statistical Association},
  volume={67},
  number={338},
  pages={429--431},
  year={1972},
  publisher={Taylor \& Francis}
}

@article{berman1973local,
  title={Local nondeterminism and local times of Gaussian processes},
  author={Berman, Simeon M. and Getoor, Ronald},
  journal={Indiana University Mathematics Journal},
  volume={23},
  number={1},
  pages={69--94},
  year={1973},
  publisher={JSTOR}
}

@article{cuzick1982joint,
  title={Joint continuity of Gaussian local times},
  author={Cuzick, Jack and DuPreez, Johannes P.},
  journal={The Annals of Probability},
  pages={810--817},
  year={1982},
  publisher={JSTOR}
}

@article{xiao2007strong,
  title={Strong local nondeterminism and sample path properties of Gaussian random fields},
  author={Xiao, Yimin},
  journal={Asymptotic theory in probability and statistics with applications},
  pages={136--176},
  year={2007},
  publisher={Higher Education Press Beijing}
}

@article{devroye1999almost,
author = {Luc Devroye and Gábor Lugosi},
title = {Almost sure classification of densities},
journal = {Journal of Nonparametric Statistics},
volume = {14},
number = {6},
pages = {675--698},
year = {2002},
publisher = {Taylor \& Francis},
doi = {10.1080/10485250215323},
URL = {https://doi.org/10.1080/10485250215323},
eprint = {https://doi.org/10.1080/10485250215323}
}

@article{waghmare2025functional,
  title={The functional graphical lasso},
  author={Waghmare, Kartik G. and Masak, Tomas and Panaretos, Victor M.},
  journal={The Annals of Statistics},
  volume={53},
  number={5},
  pages={1857--1885},
  year={2025},
  publisher={Institute of Mathematical Statistics}
}

@article{dempster1972covariance,
  title={Covariance selection},
  author={Dempster, Arthur P.},
  journal={Biometrics},
  pages={157--175},
  year={1972},
  publisher={JSTOR}
}

@article{carroll2013unexpected,
  title={Unexpected properties of bandwidth choice when smoothing discrete data for constructing a functional data classifier},
  author={Carroll, Raymond J. and Delaigle, Aurore and Hall, Peter},
  journal={Annals of statistics},
  volume={41},
  number={6},
  pages={2739},
  year={2013}
}

@article{bickel2008regularized,
author = {Peter J. Bickel and Elizaveta Levina},
title = {{Regularized estimation of large covariance matrices}},
volume = {36},
journal = {The Annals of Statistics},
number = {1},
publisher = {Institute of Mathematical Statistics},
pages = {199 -- 227},
year = {2008},
doi = {10.1214/009053607000000758},
URL = {https://doi.org/10.1214/009053607000000758}
}

@article{liu2020cholesky,
author = {Yu Liu and Zhao Ren},
title = {{Minimax estimation of large precision matrices with bandable Cholesky factor}},
volume = {48},
journal = {The Annals of Statistics},
number = {4},
publisher = {Institute of Mathematical Statistics},
pages = {2428 -- 2454},
year = {2020},
doi = {10.1214/19-AOS1893},
URL = {https://doi.org/10.1214/19-AOS1893}
}

@article{lam2009sparsistency,
 ISSN = {00905364, 21688966},
 URL = {http://www.jstor.org/stable/25662231},
 author = {Clifford Lam and Jianqing Fan},
 journal = {The Annals of Statistics},
 number = {6B},
 pages = {4254--4278},
 publisher = {Institute of Mathematical Statistics},
 title = {SPARSISTENCY AND RATES OF CONVERGENCE IN LARGE COVARIANCE MATRIX ESTIMATION},
 urldate = {2026-02-19},
 volume = {37},
 year = {2009}
}

@article{meinshausen2006graphs,
author = {Nicolai Meinshausen and Peter B{\"u}hlmann},
title = {{High-dimensional graphs and variable selection with the Lasso}},
volume = {34},
journal = {The Annals of Statistics},
number = {3},
publisher = {Institute of Mathematical Statistics},
pages = {1436 -- 1462},
year = {2006},
doi = {10.1214/009053606000000281},
URL = {https://doi.org/10.1214/009053606000000281}
}

@book{kaipio2005statistical,
  title={Statistical and computational inverse problems},
  author={Kaipio, Jari P. and Somersalo, Erkki},
  year={2005},
  publisher={Springer}
}

@article{zhou2025dynamic,
  title={Dynamic modelling of sparse longitudinal data and functional snippets with stochastic differential equations},
  author={Zhou, Yidong and Mueller, Hans-Georg},
  journal={Journal of the Royal Statistical Society Series B: Statistical Methodology},
  volume={87},
  number={3},
  pages={833--849},
  year={2025},
  publisher={Oxford University Press UK}
}
